\pdfoutput=1
\documentclass[11pt]{article}
\usepackage[margin=1in]{geometry}
\usepackage[utf8]{inputenc} 
\usepackage{booktabs}       
\usepackage{amsfonts}       
\usepackage{nicefrac}       
\usepackage{microtype}      
\usepackage[svgnames]{xcolor}
\usepackage{parskip}        
\usepackage{amsmath, amsthm, mathtools, dsfont}
\usepackage{txfonts} 
\usepackage{color,graphicx}
\usepackage{url,hyperref}
\usepackage{enumerate}
\usepackage{algorithm}
\usepackage[noend]{algpseudocode}
\usepackage{fancybox}
\usepackage{xspace}
\usepackage{cuted,tcolorbox,lipsum}
\usepackage{multicol}
\usepackage{xcolor}
\usepackage[T1]{fontenc}
\usepackage{tikz}
\usepackage{tcolorbox}
\usepackage{hyperref}
\usepackage[all]{hypcap}
\usetikzlibrary{matrix,positioning,calc,patterns}
\tcbuselibrary{skins}
\usepackage[colorinlistoftodos]{todonotes}
\usepackage{array,multirow}
\usepackage{framed}
\usepackage{mdframed}
\usepackage{amsthm}

\newcommand*\samethanks[1][\value{footnote}]{\footnotemark[#1]}

\theoremstyle{definition}

\colorlet{shadecolor}{gray!20}

\newtheorem{theorem}{Theorem}
\newtheorem{corollary}[theorem]{Corollary}
\newtheorem{lemma}[theorem]{Lemma}

\newtheorem{definition}[theorem]{Definition}

\newcommand{\PCF}{\textsc{RecursivePCF}}
\newcommand{\SSW}{{\textsc{SSW}}}
\newcommand{\finddeltae}{{\textsc{FirstTightEdgeEvent}}}
\newcommand{\finddeltaf}{{\textsc{FirstTightFamilyEvent}}}
\newcommand{\findtightfamily}{{\textsc{FindTightFamily}}}

\newcommand{\G}{\ensuremath{G}}
\newcommand{\V}{\ensuremath{V}}
\newcommand{\E}{\ensuremath{E}}
\newcommand{\family}
{\ensuremath{\mathcal{S}}}

\newcommand{\Va}{\ensuremath{V_a}}
\newcommand{\Vi}{\ensuremath{V_i}}

\newcommand{\Er}{\E_{single}}

\newcommand{\R}{R}
\newcommand{\I}{I}

\newcommand{\Opt}{\text{OPT}}
\newcommand{\Optr}{\Opt_\R}
\newcommand{\Alg}{\text{SSW}}
\newcommand{\Sol}{\text{SOL}}
\newcommand{\Solr}{\Sol_\R}

\newcommand{\Rec}{\text{REC}}

\newcommand{\D}{\mathcal{D}}
\newcommand{\Dm}{\D_m}
\newcommand{\Dopt}{\D_{\Opt}}
\newcommand{\Dpaid}{\D_{paid}}
\newcommand{\Dunpaid}{\D_{unpaid}}
\newcommand{\Dr}{\D_\R}
\newcommand{\Dalg}{\D_\Alg}
\newcommand{\Dbase}{\D_{base}}
\newcommand{\blueDbase}{{\color{Blue} \Dbase}}
\newcommand{\Mbase}{M_{base}}
\newcommand{\redMbase}{{\color{Green} \Mbase}}

\newcommand{\Dsingle}{\D_{single}}
\newcommand{\Dmult}{\D_{multiple}}

\newcommand{\cc}{\mathcal{C}\mathcal{C}}
\newcommand{\pen}{\mathcal{P}}

\newcommand{\F}{F}
\newcommand{\Fopt}{\F_{\Opt}}
\newcommand{\Falg}{\F_{\Alg}}
\newcommand{\Frec}{\F_{\Rec}}
\newcommand{\Fr}{\F_\R}
\newcommand{\Fc}{\F_{C}}
\newcommand{\y}{y}
\newcommand{\yS}{\y_{S}}
\newcommand{\ysingle}{\y_{single}}
\newcommand{\ymult}{\y_{multiple}}

\newcommand{\cost}{cost}
\newcommand{\costr}{\cost_{\R}}
\newcommand{\ce}{c}

\newcommand{\pir}{\pi_\R}
\newcommand{\closure}{\textsc{closure}}
\newcommand{\comp}[1]{#1^c}
\newcommand{\deltaS}{\delta(S)}
\newcommand{\degree}{deg}

\newcommand{\activesets}{\mathcal{C}_a}
\newcommand{\positivesets}{\mathcal{C}_+}
\newcommand{\currentsets}{\mathcal{M}}
\newcommand{\epse}{\varepsilon_{edge}}
\newcommand{\epsf}{\varepsilon_{family}}
\newcommand{\epst}{\varepsilon}

\newcommand{\redS}{\textcolor{Red}{S}}
\newcommand{\blackC}{\textcolor{Black}{C}}
\newcommand{\blueCone}{\textcolor{Blue}{C_1}}
\newcommand{\greenCtwo}{\textcolor{Green}{C_2}}

\begin{document}
\title{Prize-Collecting Forest with Submodular Penalties: Improved Approximation}
%
%
\author{
Ali Ahmadi\thanks{University of Maryland.}\\
\texttt{ahmadia@umd.edu}
\and
Iman Gholami\samethanks\\
\texttt{igholami@umd.edu}
\and
MohammadTaghi Hajiaghayi\samethanks\\
\texttt{hajiagha@umd.edu}
\and 
Peyman Jabbarzade\samethanks\\
\texttt{peymanj@umd.edu}
\and
Mohammad Mahdavi\samethanks\\
\texttt{mahdavi@umd.edu}
}
%
%
\maketitle              
\begin{abstract}
    \emph{Constrained forest problems} form a class of graph problems where specific connectivity requirements for certain cuts within the graph must be satisfied by selecting the minimum-cost set of edges.
    The prize-collecting version of these problems introduces flexibility by allowing penalties to be paid to ignore some connectivity requirements.

    Goemans and Williamson~\cite{DBLP:journals/siamcomp/GoemansW95} introduced a general technique and developed a 2-approximation algorithm for constrained forest problems. 
    Further, Sharma, Swamy, and Williamson~\cite{DBLP:conf/soda/SharmaSW07} extended this work by developing a 2.54-approximation algorithm for the prize-collecting version of these problems. Motivated by the generality of their framework, which includes problems such as Steiner trees, Steiner forests, and their variants, we pursued further exploration.

    We present a significant improvement by achieving a 2-approximation algorithm for this general model, matching the approximation factor of the constrained forest problems. 

\end{abstract}
%
%
%

\section{Introduction}
In the context of connectivity problems, many algorithms focus on satisfying specific demands while minimizing costs. 
These problems typically involve weighted graphs, where we select a subset of edges to fulfill these demands. 
The Minimum Spanning Tree (MST) is a well-known example where efficient algorithms can find the optimal solution. 
However, numerous generalizations of MST are NP-hard, meaning finding guaranteed optimal solutions becomes computationally intractable.
Therefore, for these problems, researchers have shifted their focus to developing effective approximate solutions. 
This work addresses  constrained forest problems with submodular penalty functions for unmet demands, achieving a 2-approximation guarantee.

We formally define the problem on a weighted graph $\G = (\V, \E)$ with a weight function $\ce: \E \rightarrow \mathbb{R}_{\ge 0}$ assigning a cost to each edge. 
Additionally, we have a penalty function $\pi: 2^{2^{\V}} \rightarrow \mathbb{R}_{\ge 0}$ defined on families of subsets of vertices. 
The objective is to find a subset of edges $F \subseteq \E$ and define a family of unsatisfied sets $\family \subseteq 2^{\V}$ such that $\family$ contains all sets that do not cut $F$ (i.e., no edge in $F$ has exactly one endpoint in those sets). 
We aim to minimize the total cost, which is the sum of the cost of the selected edges, $\ce(F)$, plus the penalty for unsatisfied demands, $\pi(\family)$.
For the rest of the paper, we refer to this problem as the \emph{prize-collecting forest problem (PCF).}

The penalty function captures the cost associated with leaving certain connectivity requirements unmet. Since its domain is exponential in size, storing penalties for all inputs is impossible. Thus, we assume the availability of an oracle that, given a family of subsets of vertices, returns its penalty. This assumption is reasonable for all submodular problems, as they have an exponential domain size, and unlike additive functions, the value of any subset of elements cannot be computed as a function of the values for individual elements. 
The penalty function also satisfies the following properties:
\begin{itemize}
\item Empty set property: $\pi(\emptyset) = 0$.
\item Submodularity: $\forall \family_1, \family_2 \subseteq 2^{\V}$: $\pi(\family_1) + \pi(\family_2) \ge \pi(\family_1 \cup \family_2) + \pi(\family_1 \cap \family_2)$.
\item Monotonicity: $\forall \family_1 \subseteq \family_2 \subseteq 2^{\V}$: $\pi(\family_1) \le \pi(\family_2)$.
\item Union property: $\forall S_1, S_2 \subseteq \V$: $\pi(\{S_1, S_2\}) = \pi(\{S_1, S_2, S_1 \cup S_2\})$.
\item Complementarity: $\forall S \subseteq \V$: $\pi(\{S\}) = \pi(\{S, \comp{S}\})$.
\end{itemize}

The empty set property establishes the baseline: no penalty is incurred if all connectivity requirements are satisfied.
Submodularity indicates that having a larger set of unsatisfied sets reduces the marginal cost of not satisfying additional sets.
Monotonicity ensures that the penalty increases as the number of unsatisfied sets grows.


The last two properties, union and complementarity, relate to how the penalty treats unsatisfied sets.
The union property states that the penalty for violating the connectivity constraint on the union of two unsatisfied sets remains unchanged compared to the penalties of the individual sets. 
This is because the union inherits the lack of connectivity from its subsets.
Complementarity asserts that the penalty for a set is the same as the penalty for its complement.
This holds true because both the set and its complement cut the same edges of the graph, and failing to connect either results in the same connectivity situation.
These properties together imply that for any $\family \subseteq 2^{\V}$, we have $\pi(\family) = \pi(\closure(\family))$, where $\closure(\family)$ is the smallest superset of $\family$ that is closed under union and complement, and consequently, intersection and set difference operations.

While numerous studies focus on specific cases of our model, some of which are discussed later, the current state-of-the-art for this general problem is a \(2.54\)-approximation algorithm by Sharma, Swamy, and Williamson~\cite{DBLP:conf/soda/SharmaSW07}, which employs LP rounding. 
They also developed a \(3\)-approximation algorithm based on a primal-dual approach, which our method frequently uses as a black-box subroutine. 
Through a combinatorial analysis, we demonstrate that our approach achieves a better approximation factor by establishing tighter bounds on both the optimal solution and our algorithm’s output.

Several prior studies in connectivity research address similar problems.
Notably, Agrawal, Klein, and Ravi~\cite{DBLP:journals/siamcomp/AgrawalKR95} explored a generalization of Steiner forests, focusing on connecting vertex pairs with $r$ paths.
Building on this work, Goemans and Williamson~\cite{DBLP:journals/siamcomp/GoemansW95} introduced a powerful technique that yielded constant-factor approximation algorithms for various special cases of our problem, inspiring numerous subsequent studies including ours. 
Their model leverages vertex subsets to define connectivity demands. 
More precisely, they define a function $f: 2^{\V} \rightarrow \{0,1\}$ such that for any subset $S \subseteq \V$, $f(S)=1$ if there is a requirement to connect vertices within $S$ to those outside, and $f(S)=0$ otherwise. 
We extend their model by assigning penalties to each subset of vertices , allowing these penalties to be zero if connecting the set's vertices to the rest of the graph is not required (i.e., $f(S)=0$). 
Additionally, beyond prize-collecting, there are other generalizations where the function \( f \) can exceed \( 1 \) (see \cite{DBLP:conf/stoc/WilliamsonGMV93,DBLP:journals/combinatorica/Jain01}). This means that a minimum number of edges is mandated for each cut.

Expanding upon the work of Goemans and Williamson~\cite{DBLP:journals/siamcomp/GoemansW95}, Sharma, Swamy, and Williamson~\cite{DBLP:conf/soda/SharmaSW07} are not the only ones who have considered the prize-collecting version of the generalized forest problem. 
Other works~\cite{DBLP:conf/soda/HayrapetyanST05,DBLP:conf/waoa/NagarajanSW08,DBLP:conf/ipco/HajiaghayiKKN10,DBLP:conf/soda/BateniCEHKM11} explore variations of the prize-collecting version, differing in the definition of the penalty function and the connectivity constraint.
For example,~\cite{DBLP:conf/waoa/NagarajanSW08} study a generalization of~\cite{DBLP:conf/soda/SharmaSW07} with a $K$ connectivity requirement. 
They use the algorithm from~\cite{DBLP:conf/soda/SharmaSW07} as a black box and provide a $2.54\log(K)$ approximation algorithm.
However, since the primal-dual properties of~\cite{DBLP:conf/soda/SharmaSW07} are not trivially observed in our algorithm, a key question is whether it is possible to improve the result of~\cite{DBLP:conf/waoa/NagarajanSW08} and achieve a $2\log(K)$ approximation algorithm for their problem.

One of the well-known special cases of our model is the Steiner tree problem, where the goal is to connect a designated set of vertices. 
This problem has an obvious 2-approximate solution.  
Many works have addressed Steiner tree~\cite{DBLP:journals/algorithmica/Zelikovsky93,DBLP:journals/jco/KarpinskiZ97,DBLP:journals/siamdm/RobinsZ05}, leading to a
$(\ln{4}+\epsilon)$-approximation algorithm~\cite{DBLP:conf/stoc/ByrkaGRS10}.
Similarly, for the prize-collecting Steiner tree problem, where a penalty must be paid for each unconnected vertex, studies began with a 3-approximate solution by ~\cite{DBLP:journals/mp/BienstockGSW93}, continued with a 2-approximation algorithm by~\cite{DBLP:journals/siamcomp/GoemansW95}, and then~\cite{DBLP:journals/siamcomp/ArcherBHK11} broke the 2-approximation barrier. Recently,~\cite{DBLP:conf/stoc/AhmadiGHJM24} achieved a 1.79-approximation algorithm.
These advancements in approximating special cases of our model, all achieving factors below 2, naturally raise the question: can we develop an algorithm with an approximation factor strictly less than 2 for our general problem?

The Steiner forest problem, a generalization of Steiner tree, considers a set of demands where each demand specifies a set of vertices that must be connected.
The 2-approximation algorithm by~\cite{DBLP:journals/siamcomp/AgrawalKR95} remained unimproved for over three decades until a recent breakthrough by~\cite{ahmadi2025breaking}, which provided a slight improvement.
This suggests that better-than-2 approximations may be possible for more general problems.
The techniques introduced in~\cite{ahmadi2025breaking} offer promising new directions for improving approximation algorithms in our broader model.

It is worth noting that the prize-collecting version of the Steiner forest (PCSF), another special case of our model, had a 2.54-approximation algorithm~\cite{DBLP:conf/soda/HajiaghayiJ06} until a recent work by Ahmadi, Gholami, Hajiaghayi, Jabbarzade, and Mahdavi~\cite{10.1145/3722551} achieved a 2-approximation for this problem. 
This development opens doors to potential improvements in our model's approximation factor as well.
Our algorithm leverages their general recursive approach, a powerful tool for designing algorithms in connectivity problems (e.g., ~\cite{DBLP:conf/stoc/AhmadiGHJM24}). 
Additionally, we demonstrate the applicability of their approach to a more general model.

Compared to PCSF problem, designing an algorithm for our model introduces significantly greater complexity.
One major source of this complexity is the penalty structure. In PCSF, penalties are defined for \( n^2 \) elements, while in our model, penalties are defined for \( 2^n \) elements, with the penalty function applicable to any subset of them. This leads to \( 2^{2^n} \) potential combinations, which considerably increases the computational complexity. Attempting to access or compute all penalties would result in exponential runtime, which is infeasible.

Additionally, the use of submodular penalty functions further complicates the analysis. With additive penalty functions, if both our solution and the optimal solution pay penalties for overlapping sets of elements, we can simplify the analysis by ignoring those elements, as both solutions are charged the same fixed penalty. However, with submodular functions, the penalty for a set of elements is not independent; it depends on the presence of other elements in the penalty calculation. This dependency implies that the penalty for a given set of elements may vary between our solution and the optimal solution, based on which other elements are penalized. Consequently, a more refined analysis is required to account for these interactions accurately.



\subsection{Overview of Contribution}

In this paper, we present a 2-approximation algorithm for \emph{prize-collecting forest problem (PCF)}, formally proving the following theorem:

\begin{theorem}
    There exists a deterministic polynomial-time 2-approximation algorithm for the \emph{prize-collecting forest problem}
    with a submodular penalty function defined on the family of subsets of vertices for unsatisfied connectivity requirements.
\end{theorem}

We use the 3-approximation algorithm of~\cite{DBLP:conf/soda/SharmaSW07} as a black box, explained in Section~\ref{sec:3-apx-alg}. 
Our algorithm, described in Section~\ref{sec:2_apx_alg}, is inspired by the iterative approach of the 2-approximation algorithm for prize-collecting Steiner forests~\cite{10.1145/3722551}.
However, our model introduces additional complexities due to the submodular penalty function and exponential connectivity requirements, leading to a more challenging analysis compared to the special case addressed in~\cite{10.1145/3722551}.
The analysis of our algorithm is presented in Section~\ref{sec:2_apx_analysis}, where Theorem~\ref{theorem:main_theorem} demonstrates the approximation guarantee, and Theorem~\ref{theorem:main_time_complexity} establishes the polynomial-time complexity of our algorithm.


\paragraph{Sharma, Swamy and Williamson (SSW) algorithm.}
We briefly explain the 3-approximation algorithm for the prize-collecting forest problem developed by Sharma, Swamy, and Williamson~\cite{DBLP:conf/soda/SharmaSW07}. 
We refer to this algorithm as \SSW{} and will call it multiple times in our algorithm.

The algorithm operates on a graph where edges are assumed to be curves with lengths equal to their costs.
It starts with an empty forest, and some of its connected components are chosen as active sets.
At each moment, active sets gradually color portions of uncolored edges that have exactly one endpoint within them.

There are two events that trigger changes in the algorithm:
\begin{itemize}
    \item An edge becomes fully colored and is added to the forest, merging its endpoints' connected components. The resulting new component becomes an active set.
    \item A family of subsets of vertices becomes tight, meaning the total duration of their activation equals the penalty for violating connectivity requirements for the entire family. 
    The algorithm deactivates these sets (now called tight sets) as paying the penalty becomes more cost-effective than further coloring edges with those sets.
\end{itemize}

Once there are no more active sets left, the algorithm enters a pruning phase.
Here, it examines tight sets. 
While tight sets stop coloring their adjacent edges, other active sets may continue to color these edges and eventually fully color them, leading to their inclusion in the forest.
If a tight set cuts exactly one edge in the forest, that edge is removed. 
Since the penalty for the tight set will be paid anyway, there is no need to pay the additional cost of including the edge. 
However, if a tight set cuts multiple edges, they are left untouched. 
This is because vertices outside the tight set might rely on those edges, along with a path within the tight set, to connect to each other.

Finally, the algorithm returns the remaining forest (the fully colored edges that were not removed during pruning) and pays the penalty associated with all tight sets. 
The key idea is that all sets not cutting any edges in the forest belong to the closure of the tight sets. 
This ensures that paying the penalty of tight sets along with the obtained forest is sufficient to satisfy all connectivity requirements.

\paragraph{Our Iterative Algorithm.} 
We propose an iterative algorithm inspired by~\cite{10.1145/3722551} that calls \SSW{} on a series of newly defined penalty functions. 
We begin by calling \SSW{} and consider its resulting forest as the initial solution. 
We also assume \SSW{} provides the family of tight sets it identified (denoted by $\Dalg$).
For the rest of our algorithm, we pay the penalty of $\Dalg$, and since its penalty is already paid, we do not need to satisfy the connectivity of sets in $\closure(\Dalg)$.

A key component of our approach is a newly defined penalty function, denoted by $\pir:2^{2^{\V}} \rightarrow \mathbb{R}_{\ge 0}$. 
The value,  $\pir(\family)$, represents the marginal cost of adding a new family $\family$ to the existing paid sets from SSW, $\Dalg$. 
Formally, $\pir(\family) = \pi(\family \cup \Dalg) - \pi(\Dalg)$.  
Since $\pir$ inherits the desired properties of a penalty function, we can recursively call our algorithm using $\pir$ as the new penalty function.  
This recursive call involves running \SSW{} with $\pir$ and then creating another penalty function based on the new tight sets identified by that \SSW{} call.

The recursion continues until a call to \SSW{} results in a forest with zero penalty for tight sets. 
This indicates that the newly created penalty function in this step is identical to the previous one, and further recursion would not yield improvement.
Once this termination condition is met, we stop the recursion and return the best forest encountered throughout the iterative process.

\paragraph{Analysis.} 
First, note that to prove the approximation guarantee of the \SSW{} algorithm, it can be shown that:
\begin{itemize}
    \item The optimal solution is at least equal to the total coloring duration of sets since every set colors at least one edge of the optimal solution at each moment, or the optimal solution pays penalty proportional to that amount of coloring.
    \item The forest \SSW{} returns costs at most twice the total coloring duration since, on average, every set colors at most two edges of the forest.
    \item The penalty \SSW{} pays is at most equal to the total coloring duration as \SSW{} only pays for tight sets whose coloring duration is equal to their penalty.
\end{itemize}
This information is sufficient to conclude that \SSW{} returns a 3-approximate solution.

For the analysis of our algorithm, we leverage its recursive nature by employing induction on the recursive steps. 
In the base case where $\pi(\Dalg) = 0$, we demonstrate that the \SSW{} algorithm returns a $2$-approximate solution, which follows trivially from the earlier analysis.
For our induction hypothesis, we assume that the recursive call provides a $2$-approximate solution for the instance with the penalty function $\pir$. 
In the induction step, we need to demonstrate that the best solution, chosen between the one returned by the recursive call and the initial solution provided by \SSW{}, is a $2$-approximate solution for the initial instance with the penalty function $\pi$. 

To prove the induction step, we focus on $\Dalg$. Since we pay the penalty of $\Dalg$, it's crucial to ensure that paying the penalty for each set in $\Dalg$ leads to a good solution, whether from the recursive call or the initial \SSW{} algorithm. We partition sets in $\Dalg$ into $\Dpaid$ and $\Dunpaid$ based on whether the optimal solution pays their penalty or connects them, respectively. Sets in $\Dunpaid$ are further divided into $\Dsingle$ and $\Dmult$ based on whether they cut exactly one edge or multiple edges of the optimal solution, respectively. Now, let's focus on each class separately.

\begin{itemize}
    \item For sets in $\Dpaid$, since the optimal solution pays their penalty, it's logical for our algorithm to do the same.
    \item  Regarding sets in $\Dmult$, since they cut more than one edge of the optimal solution, a better lower-bound for the optimal solution can be found in the analysis of the \SSW, resulting in an improved approximation from \SSW's output.
    \item As for sets in $\Dsingle$, which only cut one edge of the optimal solution, we demonstrate that removing these edges from the optimal solution does not increase its penalty in the recursive instance where the penalty of $\Dalg$ has already been paid.
    Thus, a solution for the recursive instance can be achieved by removing these edges from the initial instance's optimal solution, resulting in a significantly reduced cost due to the removal of edges.
    By leveraging the induction hypothesis that the recursive call provides a 2-approximate solution for its instance, we conclude that the recursive call yields a good solution for the initial instance.
\end{itemize}




\subsection{Preliminaries}
Here, we introduce some notation and elementary lemmas used in the subsequent sections. 

\paragraph{Penalty Function Properties.}
\begin{definition}
    For a family $\family \subseteq 2^\V$, we use $\closure(\family)$ to denote the smallest subset of $2^\V$ that contains $\family$ and is closed under union and complement operations. 
    That is, for any sets $S_1, S_2 \in \closure(\family)$, $S_1\cup S_2$ and $\comp{{S_1}}$ are also in $\closure(\family)$. 
    It is worth noting that $\closure(\family)$ is closed under intersection and set difference operations as well, since these operations can be written as a combination of union and complement operations.
\end{definition}

Unless otherwise specified, when a penalty function $\pi$ is used hereinafter, we assume the function satisfies the required properties: empty set property, submodularity, monotonicity, union property, and complementarity. 

\begin{lemma} \label{lm:closure}
    For any family $\family$, sets $S_1, S_2 \in \family$, and family $\family'$ with $\pi(\family')=0$ we have 
    $$\pi(\family)=\pi(\family\cup\{S_1\cup S_2\})=\pi(\family\cup\family') = \pi(\family\cup \{S_1^C\}).$$ Consequently, we have $\pi(\family)=\pi(\closure(\family))$.
\end{lemma}
\begin{proof}
    These statements result from the required properties of $\pi$. For example, we have
    \begin{align*}
        \pi(\family \cup \{S_1\cup S_2\}) &\leq \pi(\family \setminus \{S_1\cup S_2\}) + \pi(\{S_1,S_2,S_1\cup S_2\}) - \pi(\{S_1, S_2\}) \tag{Submodularity}\\
        &=\pi(\family \setminus \{S_1\cup S_2\}) \tag{Union property}\\
        &\leq \pi(\family). \tag{Monotonicity}
    \end{align*}
    By monotonicity, we also have $\pi(\family) \leq \pi(\family \cup \{S_1\cup S_2\})$, so $\pi(\family) = \pi(\family \cup \{S_1\cup S_2\})$. The other statements can be proven similarly.
\end{proof}

\begin{definition} \label{def:marg}
We use the following notation to denote the marginal cost of the penalty function $\pi$ in relation to a given family $\family_2$:
$$
\pi(\family_1\mid \family_2) \coloneqq \pi(\family_1 \cup \family_2) - \pi(\family_2).
$$
\end{definition}

We can show that this marginal cost function retains the required properties of the original cost function $\pi$.

\begin{lemma} \label{lm:margsub}
    If $\pi$ satisfies the required properties (empty set property, submodularity, monotonicity, union property, and complementarity), then for any fixed family $\family$, the function $\pi(\ \cdot\mid\family)$ will also satisfy these properties. 
\end{lemma}
\begin{proof}
We briefly show how each property is satisfied.
\begin{itemize}
    \item Empty set property: $\pi(\emptyset\mid\family)=\pi(\emptyset\cup\family)-\pi(\family)=\pi(\family)-\pi(\family)=0$.
    \item Submodularity: For any $\family_1,\family_2 \subseteq 2^\V$ we have
    \begin{align*}
    \pi(\family_1\mid\family) + 
    \pi(\family_2\mid\family) &= 
    \pi(\family_1\cup\family) + 
    \pi(\family_2\cup\family) - 2\pi(\family)\\
    &\geq  \pi(\family_1\cup\family_2\cup\family) + 
    \pi((\family_1\cap\family_2)\cup\family) - 2\pi(\family) \tag{By submodularity of $\pi$}\\
    &\geq \pi(\family_1\cup\family_2\mid\family) + \pi(\family_1\cap\family_2\mid\family).
    \end{align*}
    \item Monotonicity: For any $\family_1\subseteq\family_2\subseteq2^\V$, we have
    $$\pi(\family_1\mid\family)=\pi(\family_1\cup\family)-\pi(\family)\leq\pi(\family_2\cup\family)-\pi(\family)=\pi(\family_2\mid\family).$$
    \item Union property: For any two sets $S_1,S_2 \subseteq \V$, we have
    \begin{align*}
    \pi(\{S_1, S_2\}\mid\family)&=\pi(\{S_1, S_2\}\cup\family)-\pi(\family)\\
    &=\pi(\{S_1, S_2\}\cup\{S_1\cup S_2\}\cup\family)-\pi(\family) \tag{Lemma \ref{lm:closure}}\\
    &=\pi(\{S_1, S_2, S_1\cup S_2\}\mid\family)
    \end{align*}
    \item Complementarity: For any set $S \subseteq \V$, we have
    \begin{align*}
    \pi(\{S\}\mid\family)&=\pi(\{S\}\cup\family)-\pi(\family)\\
    &=\pi(\{S\}\cup\{\comp{S}\}\cup\family)-\pi(\family) \tag{Lemma \ref{lm:closure}}\\
    &=\pi(\{S,\comp{S}\}\mid\family)
    \end{align*}
\end{itemize}
\end{proof}

Next, we show that the marginal cost of a marginal cost of a penalty function can be written as one marginal cost.

\begin{lemma} \label{lm:union_marginal_cost}
    For penalty functions $\pi, \pi', \pi''$, and families $\family_1, \family_2$, such that for any family $\family$, we have $\pi'(\family) = \pi(\family \mid \family_1)$, and $\pi''(\family) = \pi'(\family \mid \family_2)$, we can conclude that
    $$\pi''(\family) = \pi(\family \mid \family_1 \cup \family_2)\text{.}$$
\end{lemma}
\begin{proof}
    By definition, we have
    \begin{align*}
        \pi''(\family) &= \pi'(\family \mid \family_2)\\
        &=\pi'(\family\cup\family_2) - \pi'(\family_2)\\
        &=\pi(\family\cup\family_2\mid \family_1) - \pi(\family\mid\family_1)\\
        &=\left[\pi(\family\cup\family_1\cup\family_2) - \pi(\family_1)\right] - 
        \left[\pi(\family_1\cup\family_2) - \pi(\family_1)\right]\\
        &=\pi(\family\cup\family_1\cup\family_2) - \pi(\family_1\cup\family_2)\\
        &=\pi(\family\mid\family_1\cup\family_2).
    \end{align*}
\end{proof}

We also use the following well-known properties of submodular and montone submodular functions.

\begin{lemma} \label{lm:subadd}
    Any non-negative submodular function $\pi$ is also subadditive. That is, for any two sets $A$ and $B$, $$\pi(A\cup B) \leq \pi(A) + \pi(B).$$ 
    In particular, for any set $S$, we have $$
    \pi(S) \leq \sum_{s\in S} \pi(\{s\}).$$ 
\end{lemma}
\begin{proof}
    This follows from the definition of submodularity for any non-negative submodular function:
    \begin{align*}
    \pi(A\cup B) &\leq \pi(A) + \pi(B) - \pi(A\cap B) \tag{Submodularity}\\
    & \leq \pi(A) + \pi(B). \tag{Non-negativity}
    \end{align*}
\end{proof}

\begin{lemma} \label{lm:submod}
    If $\pi$ is a monotone submodular function, then for any three sets $S$, $A$, and $B$ such that $A \subseteq B$, we have $$\pi(S\mid B) \leq \pi(S\mid A).$$
\end{lemma}
\begin{proof}
    We use the definition of submodularity with sets $S\cup A$ and $B$ as follows:
    \begin{align*}
        \pi(S\cup A \cup B)  + \pi((S\cup A) \cap B) \leq \pi(S\cup A) + \pi(B).
    \end{align*}
    Rearranging the terms, we get 
    $$
    \pi(S\cup A \cup B) - \pi(B)\leq \pi(S\cup A) - \pi((S\cup A) \cap B),$$
    which we can use to complete the proof:
    \begin{align*}
            \pi(S\mid B) &=
            \pi(S \cup B) - \pi(B)\\
            &=\pi(S\cup A \cup B) - \pi(B) \tag{$A\subseteq B$}\\ &\leq \pi(S\cup A) - \pi((S\cup A) \cap B)\\
            &= \pi(S\cup A) - \pi((S\cap B) \cup A) \tag{$A \subseteq B$}\\
            &\leq \pi(S\cup A) - \pi(A) \tag{Monotonicity}\\
            &=\pi(S\mid A).
    \end{align*}
\end{proof}

\paragraph{Graph Notations.}
For a set $S \subseteq \V$, we use $\delta(S)$ to denote the set of edges cut by $S$, i.e., edges with exactly one endpoint in $S$. 
We also say that an edge $e\in \E$ and set $S \subseteq \V$ are adjacent if and only if $e\in \deltaS$.
For a forest $\F$, we use $\cc(\F)$ to denote the connected components of $\F$. For a subgraph $H$ of graph $\G$, we use $\ce(H)$ to denote the total weight $\sum_{e\in H} \ce(e)$ of this subgraph.

\begin{definition}
\label{def:pen}
    For a forest $\F$, we use 
    $$
    \pen(\F)\coloneqq \{S \subseteq \V \mid \lvert\delta(S)\cap\F\rvert=0\}$$
    to denote the family for which the penalty $\pi$ needs to be paid.
\end{definition}
\begin{lemma} \label{lm:compcl}
    For any forest $\F$, we have 
    $$\pen(\F)=\closure(\cc(\F)).$$
\end{lemma}
\begin{proof}
    Consider any set $S$ in $\pen(\F)$. This set must not cut any edge in $\F$. Therefore, for each component $C$ in $\F$, $S$ must either completely contain $C$ or have no intersection with $C$ as otherwise, $S$ would cut an edge in $\F$. This means that $S$ can be written as a union of sets in $\cc(\F)$ and therefore $\pen(\F) \subseteq \closure(\cc(\F))$. On the other hand, any connected component of $\F$ cuts no edge in $\F$, so $\cc(\F)\subseteq \pen(\F)$. Additionally, it is easy to see that $\pen(\F)$ is closed under union and complement operations: if a set cuts no edge in $\F$, neither will its complement, and if two sets cut no edges in $\F$, neither will their union. Therefore, $\closure(\cc(\F)) \subseteq \pen(\F)$. 
\end{proof}

For an instance $\I=(\G=(\V,\E,\ce),\pi)$ of the PCF problem and a solution $\Sol$ using forest $\F_\Sol$, we use 
$\cost_\I(\Sol)=\ce(\F_\Sol)+\pi(\pen(\F_\Sol))$ to denote the cost of this solution.

\section{SSW: 3-Approximation Algorithm}
\label{sec:3-apx-alg}

In this section, we present the algorithm of~\cite{DBLP:conf/soda/SharmaSW07} for the sake of completeness. 
Our presentation remains faithful to their original work, with the key difference being the omission of the function $f(S)$. 
This function determines whether the vertices in set $S$ need to be connected to outer vertices $(\V \setminus S)$.
Following~\cite{DBLP:conf/soda/SharmaSW07}, we assume that any set $S$ with $f(S)=0$ has a zero marginal cost to the penalty of any family (refer to their inactivity property). 
Here, we argue that $f$ can be entirely removed. 
Similar to their approach, we assume that if a set $S$ does not require connections to $\V \setminus S$, its contribution to the penalty of any family is zero.
This simplified model retains the generality of their original model, as any instance compatible with their $f(S)$ function can be readily defined within our framework.

We now proceed to explain their algorithm without relying on $f$. 
The detailed pseudocode can be found in Algorithm~\ref{alg:pcf3}.
The algorithm operates on a graph where each edge $e \in \E$ is assumed to be a curve with length $\ce(e)$.
It begins with an empty forest, $\F$, and a set containing all connected components of $\F$, denoted by $\currentsets$. 
Additionally, it defines $\yS$ to track the total time each set $S \subseteq \V$ has been active throughout the algorithm's execution.

Throughout the algorithm, we maintain a subset of the connected components called active sets, denoted by $\activesets$. 
Initially, the active sets are identical to the set of all connected components, denoted by $\currentsets$. 
These active sets play a key role in the algorithm.
At each step, an active set colors the uncolored portions of its adjacent edges. 
This coloring process continues until an edge, $e \in \E$, becomes fully colored (meaning the sum of the activation times of all sets adjacent to it, $\sum_{S: e \in \deltaS} \yS$, equals the edge's length, $\ce(e)$). 
We use the procedure \finddeltae, which is explained in more detail in Section~\ref{sec:SSW_suprocedures}, to determine the earliest time when a new edge becomes fully colored.

When an edge is fully colored, the algorithm updates the forest and active sets simultaneously. 
If the endpoints belong to different components in the current forest, $\F$, the edge is added to $\F$. 
Additionally, the connected components containing both endpoints are merged in both $\currentsets$ (all connected components) and $\activesets$ (active sets). 
This merging ensures that the edge becomes part of a single connected component within the forest, preventing it from being adjacent to any active set again. 
Consequently, the sum of activation times for that edge, $\sum_{S: e \in \deltaS} \yS$, will not increase further. 
This observation leads to the following corollary.

\begin{corollary}
\label{corollary:edge_constraint}
    For any edge $e \in \E$, in the \SSW{} procedure, we have $\sum_{S: e \in \deltaS} \yS \leq \ce(e)$.
\end{corollary}

While the core algorithm resembles the general approach of Goemans and Williamson~\cite{DBLP:journals/siamcomp/GoemansW95}, there is an additional constraint we need to consider. 
For any family $\family \subseteq 2^{\V}$, the algorithm ensures that the combined activation time of all sets within a family, $\sum_{S \in \family} \yS$, never exceeds the penalty of the family, $\pi(\family)$. 

To enforce this constraint, the algorithm utilizes two procedures.
First, it calls \finddeltaf{} to identify the earliest point in time when a family containing at least one active set becomes tight, meaning its total activation time reaches the penalty of the family.
Then, it uses \findtightfamily{} to find such a family and deactivate all sets in the family.
More details on these procedures can be found in Section~\ref{sec:SSW_suprocedures}.
By using these procedures, the algorithm ensures that no active set is part of a tight family, establishing the following corollary.

\begin{corollary}
\label{corollary:family_constraint}
    For any family  $\family \subseteq 2^{\V}$, in the \SSW{} procedure, we have $\sum_{S \in \family} \yS \leq \pi(\family)$.
\end{corollary}

Finally, when there are no active sets remaining, the algorithm starts the pruning phase. 
Let $\Dm$ be the union of all tight families found during the algorithm.
In the pruning phase, while there is a set $S \in \Dm$ which cuts the forest in exactly one edge, that edge will be removed from the forest.
This process continues until there are no such sets.
The remaining forest $\F$ will be returned as the solution, and the penalty of $\pen(\F)$ will be paid. 
We also return $\Dm$ since our 2-approximation algorithm will use it in the next section.

\begin{algorithm}[ht]
  \caption{A 3-approximation Algorithm}
  \label{alg:pcf3}
  \hspace*{\algorithmicindent} \textbf{Input:} An undirected graph $\G=(\V, \E, \ce)$ with edge costs $\ce: \E \rightarrow \mathbb{R}_{\ge 0}$ and penalties $\pi : 2^{2^{\V}} \rightarrow \mathbb{R}_{\ge 0}$. \\
  \hspace*{\algorithmicindent} \textbf{Output:} A forest $\F$ and a family $\Dm \subseteq 2^{\V}$.
  \begin{algorithmic}[1]
    \Procedure{\SSW}{$I=(\G, \pi)$}
      \State Initialize $\F, \Dm \gets \emptyset$
      \State Initialize $\currentsets, \activesets \gets \{\{v\}: v \in \V\}$
      \State Implicitly set $\yS \gets 0$ for all $S \subset \V$
      \While{$\activesets \neq \emptyset$} 
      \label{line:while-loop-for-coloring}
        \State $\epse \gets \finddeltae(\G, \y, \activesets, \currentsets) $ 
        \State $\epsf \gets \finddeltaf(\G, \pi, \y, \activesets) $ 
        \State $\epst \gets \min(\epse, \epsf)$
        \For{$S \in \activesets$}
            \State $\yS \gets \yS + \epst$ 
        \EndFor
        \For{$e=\{u,v\}\in E$} 
          \State Let $S_u, S_v \in \currentsets$ be sets that contains each endpoint of $e$
          \If{$\sum_{S: e \in \deltaS} \yS = \ce(e)$ \textbf{and} $S_u \neq S_v$} 
            \State $\F \gets \F \cup \{e\}$
            \State $\currentsets \gets (\currentsets \setminus \{S_u, S_v\}) \cup \{S_u \cup S_v\}$ \label{line:add_to_currentsets}
            \State $\activesets \gets (\activesets \setminus \{S_u, S_v\}) \cup \{S_u \cup S_v\}$ \label{line:add_to_activesets}
          \EndIf
        \EndFor
        \State $\family \gets \findtightfamily(\G, \pi, \y, \activesets)$
        \If{$\family \neq \emptyset$}
          \State $\activesets \gets \activesets \setminus \family$ \label{line:deactivate_tight_family}
          \State $\Dm \gets \Dm \cup \family$ \label{line:insert_to_Dm}
        \EndIf
      \EndWhile
      \While{$\exists S\in \Dm, e \in \E: \deltaS \cap \F = \{e\}$} \label{line:while_loop_for_removing_edges}
        \State $\F \gets \F \setminus \{e\}$ \label{line:remove_edge}
      \EndWhile
      \State \Return $(\F, \Dm)$
    \EndProcedure
  \end{algorithmic}
\end{algorithm}


\subsection{Sub-procedures Used in \SSW{}}
\label{sec:SSW_suprocedures}

In this section, we explore the procedures \finddeltae, \finddeltaf, and \findtightfamily{} in more detail.

\paragraph{\finddeltae:} 
The \finddeltae{} procedure identifies the earliest time an edge adjacent to an active set will be fully colored. 
It examines each edge in the graph. 
If the endpoints are in separate components ($\currentsets$) and at least one belongs to an active set ($\activesets$), it calculates a value,  $\epse(e) = (\ce(e) - \sum_{S:e\in \deltaS} \yS)/ t_e$. 
This value, $\epse(e)$, reflects the remaining coloring effort for the edge. 
It considers both the edge's length, $\ce(e)$, and how much it's been colored by active sets (sum of activation times, $\sum_{S:e\in \deltaS} \yS$). 
The term $t_e \in \{1,2\}$ represents the number of endpoints belonging to active components, since more active endpoints can speed up coloring. 
Finally, \finddeltae{} returns the minimum $\epse(e)$ across all such edges, demonstrating the earliest time that an edge becomes fully colored (denoted by $\epse$).

\paragraph{\finddeltaf:}
This procedure plays a crucial role in ensuring the algorithm respects family constraints. 
Its objective is to identify the earliest time a family containing at least one active set becomes tight, meaning the total activation time of its member sets reaches a predefined threshold, $\pi(\family)$.

The procedure starts by formulating the following LP to find a threshold value, denoted by $\epsf$. 

\begin{align*}
    &\text{maximize } \epsf \\
    &\text{s.t. } \epsf \cdot |\family \cap \activesets| + \sum_{S\in \family} \yS \le \pi(\family) &\forall \family \subseteq 2^{\V} \text{.}
\end{align*}

Regarding the constraint of the LP, we can focus solely on families containing active sets (denoted by $\activesets$) and sets with positive activation time $\yS > 0$ (denoted by $\positivesets$), as these are the only sets that contribute to reaching the tightness threshold. Inactive sets with zero activation time do not affect the left-hand side of the constraint but may only increase the right-hand side due to the monotonicity of the penalty function. Therefore, we can remove them to tighten the constraint, thus simplifying our consideration.  

Even after narrowing down to relevant sets, the number of constraints in the LP remains exponential. To address this, a binary search on $\epsf$ can be employed. However, binary search on a real value may never terminate. Section 5.3 of~\cite{DBLP:conf/soda/SharmaSW07} addresses this issue by continuing the process and reducing the binary search interval until it becomes sufficiently small, leaving only one candidate for $\epsf$. This approach works because they demonstrate that $\epsf$ must be a rational number with a bounded numerator and denominator.  

Nevertheless, this step is unnecessary, as an idea similar to their binary search provides a strong separation oracle, allowing the use of the ellipsoid method to solve the LP in polynomial time instead of relying on binary search. This approach, established by Khachiyan~\cite{DBLP:journals/jal/AspvallS80}, guarantees polynomial-time solvability for LPs.  

In our strong separation oracle, for each value of $\epsf$ considered during the binary search, we need to verify whether all constraints are satisfied or find one that is violated. This translates to checking if a newly defined function,  
\[
\pi'(\family) = \pi(\family) - \left(\epsf \cdot |\family \cap \activesets| + \sum_{S \in \family} \yS\right),
\]  
remains non-negative for all relevant families, or finding a family $\family$ for which $\pi'(\family) < 0$.  

The key insight here is that both $\pi(\family)$ and the subtracted term are submodular functions. This property allows us to efficiently find the family with the minimum $\pi'(\family)$ value using submodular minimization algorithms, which are solvable in polynomial time as shown in~\cite{DBLP:journals/jct/Schrijver00,DBLP:journals/jacm/IwataFF01}. By checking whether this minimum value is less than zero, we can determine whether a constraint is violated for the current $\epsf$ and identify the violating family and its corresponding constraint.

\paragraph{\findtightfamily:}
Here, our goal is to find a tight family containing at least one active set or return the empty set if no such family exists. 
Unlike \finddeltaf, we do not require $\epsf$ and instead use the submodular function $\pi'(\family) = \pi(\family) - \sum_{S \in \family} \yS$. 
We aim to solve the submodular minimization problem for the ground set $\activesets \cup \positivesets$ to find a family $\family$ such that $\pi'(\family) = 0$. 
It's important that $\family \cap \activesets \neq \emptyset$.

To achieve this, we iterate over $\activesets$, and for each active set $S \in \activesets$, we attempt to find a tight family containing $S$. This involves defining the submodular function $\pi'_S(\family) = \pi'(\family \cup \{S\})$ and using the submodular minimization algorithm to find a family $\family$ with the minimum value of $\pi'_S(\family)$. 
If $\pi'_S(\family) = 0$, we return $\family \cup \{S\}$. 
If, however, $\pi'_S(\family) > 0$ for all $S \in \activesets$, we return $\emptyset$.


\subsection{Useful properties}
To analyze our 2-approximation algorithm, we will leverage several lemmas related to the \SSW{} algorithm, provided in this section.

Our first lemma focuses on the inequality in Corollary~\ref{corollary:family_constraint} and demonstrates that for the family $\Dm$, this inequality becomes an equality. 
This is achieved by leveraging the submodularity of the penalty function and proving that the union of two tight families remains a tight family. 
Since $\Dm$ is constructed by progressively adding sets from tight families, we can then conclude the following lemma:

\begin{lemma} \label{lm:tight}
    In procedure \SSW, the family $\Dm$ is a tight family, i.e.,
    $$\sum_{S\in \Dm} \yS = \pi(\Dm)\text{.}$$
\end{lemma}
\begin{proof}
    Let us assume that $\family_1, \family_2 \subseteq 2^{\V}$ are two tight families.
    Since $\family_1$ is tight, we have $\sum_{S\in \family_1} \yS = \pi(\family_1)$ and a similar equality for $\family_2$. 
    This implies that
    \begin{align*}
        \sum_{S\in \family_1 \cup \family_2} \yS + \sum_{S\in \family_1 \cap \family_2} \yS &= \sum_{S\in \family_1} \yS + \sum_{S\in \family_2} \yS \\
        &= \pi(\family_1)+\pi(\family_2) \tag{$\family_1$ and $\family_2$ are tight} \\
        &\ge \pi(\family_1 \cup \family_2) + \pi(\family_1 \cap \family_2)\text{.} \tag{By submodularity of $\pi$}
    \end{align*}
    However, according to Corollary~\ref{corollary:family_constraint}, we have $\sum_{S\in \family_1 \cup \family_2} \yS \le \pi(\family_1 \cup \family_2)$ and $\sum_{S\in \family_1 \cap \family_2} \yS \le \pi(\family_1 \cap \family_2)$.
    Therefore, the above inequality should be an equality and $\sum_{S\in \family_1 \cup \family_2} \yS = \pi(\family_1 \cup \family_2)$.
    
    Now, we prove the tightness of $\Dm$.
    Line~\ref{line:insert_to_Dm} is the only place that we modify $\Dm$ and add sets of a tight family $\family$ to $\Dm$.
    Hence, we have $\Dm = \bigcup_i \family_i$, where $\family_1, \family_2, \cdots$ is the sequence of tight families that are added to $\Dm$.
    Note that once these families become tight, they remain tight, as their sets are removed from the active sets in Line~\ref{line:deactivate_tight_family}, and their corresponding values $\yS$ will never change.
    Since we show that the union of two tight families is a tight family, it is easy to conclude that $\Dm$ is also a tight family.
\end{proof}

The next lemma demonstrates that selecting the forest returned by subroutine \SSW{} and paying the penalty of $\Dm$ provides a valid solution.
We prove this by showing that $\cc(\F) \subseteq \closure(\Dm)$ using induction. 
First, we show that the connected components of the forest obtained before the while loop at Line~\ref{line:while_loop_for_removing_edges} are in $\closure(\Dm)$ (induction base) since they originated from active sets that were deactivated due to inclusion in a tight family.

Next, we argue that after each edge removal in Line~\ref{line:remove_edge}, the two resulting connected components also belong to $\closure(\Dm)$ (induction step). 
We demonstrate this using the following facts:
\begin{itemize}
    \item The connected component formed before removing the edge in Line~\ref{line:remove_edge} was in $\closure(\Dm)$ (induction hypothesis).
    \item The tight set selected in Line~\ref{line:while_loop_for_removing_edges} is in $\closure(\Dm)$.
    \item The two new components created by removing $e$ in Line~\ref{line:remove_edge} can be created by applying set intersection and set difference operations on sets already known to be in $\closure(\Dm)$. Since these operations preserve closure, the resulting components are also in $\closure(\Dm)$.
\end{itemize}

Finally, since all components of final forest are in $\closure(\Dm)$ the penalty of $\pen(\F)$ is at most as large as the penalty of $\Dm$.

\begin{lemma}
\label{lm:ssw_penalty_le_Dm_penalty}
    At the end of procedure \SSW, the penalty we need to pay by selecting forest $\F$ is at most $\pi(\Dm)$, i.e.,
    $$\pi(\pen(\F)) \le \pi(\Dm)\text{.}$$
\end{lemma}
\begin{proof}
    First, we want to prove that $\cc(\F) \subseteq \closure(\Dm)$.
    Observe that at the end of procedure \SSW, within the while loop of Line~\ref{line:while_loop_for_removing_edges}, forest $\F$ creates a sequence of forests $\F_0,\F_1, \cdots, \F_k$.
    Here, $\F_0$ is the forest $\F$ before Line~\ref{line:while_loop_for_removing_edges}, $\F_i$ is obtained by removing an edge from $\F_{i-1}$ in Line~\ref{line:remove_edge}, and $k$ is the number of times an edge is removed from $\F$ in Line~\ref{line:remove_edge}.
    We will use induction to prove that $\cc(\F_k) \subseteq \closure(\Dm)$ where $\F_k$ is the forest returned by procedure \SSW.
    
    For the induction base, we need to show that $\cc(\F_0) \subseteq \closure(\Dm)$. 
    Before removing any edge we have $\currentsets = \cc(\F) = \cc(\F_0)$.
    For any set in $\currentsets$, there exists a time when it was in $\activesets$ because sets are inserted into both collections at the same time in Lines~\ref{line:add_to_currentsets} and~\ref{line:add_to_activesets}.
    Since the sets in $\currentsets$ at the end of \SSW{} have not been merged with other sets, they must have been removed from $\activesets$ in Line~\ref{line:deactivate_tight_family} as members of a tight family.
    Consequently, they were inserted into $\Dm$ in Line~\ref{line:insert_to_Dm}. 
    Therefore, we can conclude that $\cc(\F_0) = \currentsets \subseteq \Dm \subseteq \closure(\Dm)$.

Now, for the induction step, we aim to prove that for any \( i \ge 0 \), if \( \cc(\F_i) \subseteq \closure(\Dm) \), then \( \cc(\F_{i+1}) \subseteq \closure(\Dm) \). 
The following arguments establish this result:
\begin{enumerate}[i.]
    \item By the induction hypothesis, $\cc(\F_i) \subseteq \closure(\Dm)$.
    \label{item:induc-hyp}
    \item Let $S \in \Dm$ and edge $e = (v,u) \in \E$ be selected in Line~\ref{line:while_loop_for_removing_edges}. Then $\deltaS \cap \F = \{e\}$.
    \label{item:one-edge}
    \item Since $e \in \deltaS$, $S$ contains exactly one endpoint of $e$. Without loss of generality, we assume $v\in S$ and $u \notin S$.
    \item Obtain $\F_{i+1} = \F_{i} \setminus \{e\}$ through Line~\ref{line:remove_edge}.
    \label{item:fi1}
    \item Let $S'\in \cc(\F_i)$ be the component of $\F_i$ containing $e$.
    From~\eqref{item:induc-hyp} we have $S'\in \closure(\Dm)$.
    \label{item:sp-in-dm}
    \item Let $S'_1$ and $S'_2$ denote the components of $S'$ created by removing $e$, where $S'_1$ contains vertex $v$ and $S'_2$ contains vertex $u$.
    \item We have $S'_1 \subseteq S$. Otherwise, assume vertex $w \in S'_1 \setminus S$. Then $S$ cuts the path from $w$ to $v$ in an edge other than $e$, contradicting~\eqref{item:one-edge}.
    \label{item:sp1-sub-s}
    \item We have $S'_2 \cap S = \emptyset$. Otherwise, assume vertex $w \in S'_2 \cap S$. Then $S$ cuts the path from $w$ to $u$ in an edge other than $e$, contradicting~\eqref{item:one-edge}.
    \label{item:sp2-no-intersect-s}
    \item From~\eqref{item:sp1-sub-s} and~\eqref{item:sp2-no-intersect-s}, we conclude $S'_1 = S' \cap S$ and $S'_2 = S' \setminus S$.
    \label{item:s1-s2-s}
    \item Since $S, S' \in \closure(\Dm)$ (from \ref{item:one-edge} and \ref{item:sp-in-dm}), and $S'_1, S'_2$ are obtained by set intersection and set difference operations, we have $S'_1, S'_2 \in \closure(\Dm)$.
    \label{item:sp1-sp2-in-dm}
    \item From \eqref{item:fi1}, $\cc(\F_{i+1}) = (\cc(\F_{i}) \setminus \{S'\}) \cup \{S'_1, S'_2\}$.
    \label{item:ccfi1}
    \item By combining~\eqref{item:induc-hyp},~\eqref{item:sp1-sp2-in-dm}, and~\eqref{item:ccfi1}, we conclude $\cc(\F_{i+1}) \subseteq \closure(\Dm)$.
\end{enumerate}

    From the induction, we have  $\F_k \subseteq \closure(\Dm)$ where $\F_k$ is the forest $\F$ returned by \SSW.
    Now, we can complete the proof of lemma.
    \begin{align*}
        \pi(\pen(\F)) &= \pi(\closure(\cc(\F))) \tag{Lemma~\ref{lm:compcl}} \\
        &= \pi(\cc(\F)) \tag{Lemma~\ref{lm:closure}} \\
        &\le \pi(\closure(\Dm)) \tag{$\F \subseteq \closure(\Dm)$} \\
        &= \pi(\Dm) \tag{Lemma~\ref{lm:closure}}
    \end{align*}
\end{proof}

The next lemma is similar to a well-known result used in coloring problems. 
It states that the cost of the final forest returned by procedure \SSW{} is at most twice the sum of the coloring duration of each subset of vertices. 
A similar lemma can be found in~\cite{DBLP:journals/siamcomp/GoemansW95} and many other papers in similar problems.
The proof leverages the property that each edge in the final forest is fully colored, combined with the fact that, at any given moment, the number of edges being colored—those adjacent to active sets—is at most twice the number of active sets at that moment.

\begin{lemma} \label{lm:ssw_forest_cost}
    At the end of the \SSW, for the cost of the final forest $\F$, we have: 
    $$\ce(\F) \le 2\sum_{S\subseteq \V} \yS\text{.}$$
\end{lemma}
\begin{proof}
    In one step of the algorithm, we increase all \( \yS \) values by \( \epst \). During each moment of this \( \epst \) time of coloring, active sets color a portion of edges in the final forest \( \F \) equal to the number of edges in \( \F\) cut by each active set, which for set \( S \) is \( \lvert \deltaS \cap \F \rvert \).

    In addition, for this step, we define the graph constructed by contracting vertices in each component of \( \currentsets \) over the final forest \( \F \) and removing isolated vertices as \( \Fc \). $\Fc$ has the following properties:

\begin{itemize}
    \item 
    First, \( \Fc \) is a forest because \( \F \) is a forest, and \( \currentsets \) are components created by merging edges of \( \F \).

    \item
    Second, all leaves of \( \Fc \) are active sets. To see why: Assume there exists a leaf in \( \Fc \) that is not an active set at the moment. This implies that \( \Dm \) contains the corresponding set of this leaf. Let us denote it by \( S \), and the edge going out from \( S \) in \( \F \) is \( e \). We can observe that \( \deltaS \cap \F = \{e\} \) because \( S \) is a leaf in the contracted forest \( \Fc \). According to the algorithm at line \ref{line:while_loop_for_removing_edges}, \( e \) must be removed by the end of the algorithm, which leads to a contradiction.

    \item 
    For any active set $S \in \currentsets$, if $v$ is the vertex corresponding to $S$, we have $\deg_{\Fc}(v) = \lvert \deltaS \cap \F \rvert$. 
\end{itemize}

    Since all leaves in \( \Fc \) are active sets, the total degree of vertices in \( \Fc \) corresponding to active sets is at most twice the number of active sets. This can be obtained as follows:
    \begin{align*}
        \sum_{v \in \Va} \degree_{\Fc}(v) 
        &= \sum_{v \in \Va} \degree_{\Fc}(v) + \sum_{v \in \Vi} \degree_{\Fc}(v) - \sum_{v \in \Vi} \degree_{\Fc}(v) \\
        &= \sum_{v \in \Fc} \degree_{\Fc}(v) - \sum_{v \in \Vi} \degree_{\Fc}(v) \tag{$\Fc = \Va \cup \Vi, \Va \cap \Vi = \emptyset$}\\
        &\le \sum_{v \in \Fc} \degree_{\Fc}(v) - 2|\Vi|\tag{Inactive sets are not leaves or isolated}\\
        &\le 2|\Fc| - 2|\Vi|\tag{For any forest $\F$, $\sum_{v \in \F} \degree(v) \le 2|\F|$}\\
        & = 2|\Va|
    \end{align*}
    where \( \Va \) and \( \Vi \) are subsets of vertices in \( \Fc \) that correspond to active or inactive sets, respectively.

    Therefore if we translate from $\Fc$ to $\F$ we have:
    \begin{align*}
        \sum_{S \in \activesets} \lvert \deltaS \cap \F \rvert
        \le 2\lvert\activesets\rvert
    \end{align*}

    By multiplying by \( \epst \), we have:

    \[
    \sum_{S \in \activesets} \epst \cdot \lvert \deltaS \cap \F \rvert \le \epst \cdot 2 \lvert \activesets \rvert
    \]

    Here, the left-hand side represents the portion of edges colored during this $\epst$ interval, while the right-hand side represents the total coloring duration of all active sets in that same interval.

    Now, if we sum the last inequality over steps of the algorithm, the left-hand side would be at least \( \ce(\F) \) since every edge in \( \F \) receives a color. On the other hand, the right-hand side would be \( 2 \sum_{S \subseteq \V} \yS \) since in every step we increase \( \y \) by \( \epst \) for all sets in \( \activesets \). Therefore, we have completed the proof.
\end{proof}

Finally, we show that \SSW{} runs in polynomial time.
This is crucial for our algorithm's overall efficiency, as it calls \SSW{} multiple times.
We can demonstrate this by noting that the number of connected components and the number of active sets at the beginning is at most $|\V|$.
After each iteration of the while loop in Line~\ref{line:while-loop-for-coloring}, either two components are merged or one active set is deactivated.
Thus, the number of iterations of the while loop is at most $2|\V|$.
Finally, the number of operations in each iteration is polynomial, which concludes the lemma.

\begin{lemma} \label{lm:ssw_poly_time}
    The \SSW{} algorithm runs in polynomial time.
\end{lemma}
\begin{proof}
First, the while loop in line~\ref{line:while-loop-for-coloring} runs at most \( 2|\V| \) times. During each iteration, either one component becomes tight and is removed from the active sets, or two components are merged. Therefore, in at most \( |\V| \) steps, \( |\V| \) components become inactive, and there are at most \( |\V| - 1 \) merges.

Within this loop, \finddeltae{} operates in polynomial time by checking all edges, \finddeltaf{} operates in polynomial time using binary search to find the subset minimizing a submodular function. Submodular minimization is known to be achievable in polynomial time. Similarly, \findtightfamily{} operates in polynomial time.

To complete the proof, consider the second while loop in Line~\ref{line:while_loop_for_removing_edges}. Checking all subsets and edges to satisfy the loop's condition is polynomial. Additionally, since each iteration removes one edge from the forest, there are at most \( |\V| \) iterations. Hence, the total runtime is polynomial.
\end{proof}


\section{2-Approximation Algorithm for Submodular PCF}
\label{sec:2_apx_alg}
In this section, we present \PCF{}, a 2-approximation algorithm for our prize-collecting forest problem with submodular penalties, utilizing a recursive approach.
This algorithm utilizes the 3-approximation algorithm described in Section~\ref{sec:3-apx-alg} as a key component. 
Algorithm \ref{alg:pcf2} illustrates the steps of obtaining a 2-approximation solution.

Recall that the 3-approximation algorithm explained in the previous section is denoted as $\Alg$.
For a given instance $\I = (\G, \pi)$, we first call $\Alg$ (Line \ref{line:alg-output}). 
It is worth noting that we use $\Dalg$ to denote the union of tight families in the 3-approximation algorithm i.e., for instance $\I$, $\Dalg$ is the final $\Dm$ returned by Algorithm \ref{alg:pcf3}. 
If $\pi(\Dalg) = 0$, the total penalty that $\Alg$ pays is $0$ which means we have a 2-approximate solution (based on Lemma~\ref{lm:algcost}). 
Therefore, we can immediately return the output of $\Alg$ as the solution (Lines \ref{line:check-penalty-zero} and \ref{line:return-early}). 

On the other hand, since $\Alg$ suggests paying the penalty for the family $\Dalg$, we intuitively consider the penalty of this family to be already paid and solve the problem to satisfy other demands.
Since the penalty is a function of families, it is not trivial to simply subtract the penalty of these sets.
However, the \emph{marginal cost} provides a natural way to account for the cost of a family, allowing the introduction of a new penalty function.
Consequently, we define a new instance $\R = (\G, \pir)$ where $\pir$ is as follows.
For all $\family \subseteq 2^{\V}$:
\begin{align*}
    \pir(\family) &= \pi(\family \mid \Dalg).
\end{align*}
Based on Lemma~\ref{lm:margsub}, the new penalty function has the required properties for being a penalty function for our problem.
Lines \ref{line:pir} and \ref{line:construct_R} of the algorithm outline the construction of $\pir$ and the instance $\R$.

Next, we recursively find the 2-approximate solution for instance $\R$ (Line~\ref{line:get_recursive_output}). 
Assume the solution for this instance is $\Frec$. 
This solution is also valid for the input instance since the graph remains the same in instance $\R$.
The final solution is the one with the lower cost between the \SSW{} solution and the solution returned by the recursive call (Lines~\ref{line:return-min} to~\ref{line:return_Frec}). 
In the next section, we provide the analysis of the described algorithm.

\begin{algorithm}[ht]
  \caption{PCF 2-approximation algorithm}
  \label{alg:pcf2}
  \hspace*{\algorithmicindent} \textbf{Input:} An undirected graph $\G=(\V, \E, \ce)$ with edge costs $\ce: \E \rightarrow \mathbb{R}_{\ge 0}$ and penalties $\pi : 2^{2^{\V}} \rightarrow \mathbb{R}_{\ge 0}$. \\
  \hspace*{\algorithmicindent} \textbf{Output:} A forest $\F$ with the minimum cost.
  \begin{algorithmic}[1]
    \Procedure{$\PCF$}{$\I=(\G,\ \pi)$}
        \State $(\Falg, \Dalg) \gets \Alg(\I)$
        \label{line:alg-output}
        \State $\cost(\Alg) \gets \ce(\Falg) + \pi(\pen(\Falg))$
        \label{line:cost-alg}
        \If{$\pi(\Dalg)=0$} \label{line:check-penalty-zero}
        \State\Return $\Falg$ \label{line:return-early}
        \EndIf
        \State Define $\pir : 2^{2^{\V}} \rightarrow \mathbb{R}_{\ge 0}$ as $\pir(\family) = \pi(\family \mid \Dalg)$
        \label{line:pir}
        \State Let instance $\R$ of the problem be the graph $\G$ with penalty function $\pir$
        \label{line:construct_R}
        \State $\Frec \gets \PCF(\R)$ \label{line:get_recursive_output}
        \State $\cost(\Rec) \gets \ce(\Frec) + \pi(\pen(\Frec))$
        \label{line:cost-r}
        \If{$\cost(\Alg) \leq \cost(\Rec)$}
        \label{line:return-min}
        \State\Return $\Falg$
        \Else
        \State \Return $\Frec$ \label{line:return_Frec}
        \EndIf
    \EndProcedure
  \end{algorithmic}
\end{algorithm}


\section{Analysis of the 2-approximation Algorithm}
\label{sec:2_apx_analysis}

In this section, we analyze the approximation ratio achieved by Algorithm~\ref{alg:pcf2}.
To analyze our algorithm, we compare the cost of the solution returned by this algorithm on an instance $\I=(\G,\pi)$ to the cost of an optimal solution $\Opt$.
In our analysis, we focus on the algorithm's variables and output in a single $\PCF$ call where $\I$ is the input and use induction to analyze the recursive call.
Additionally, we use the variables $\yS$ from Algorithm~\ref{alg:pcf3}, executed during the call to the \SSW{} procedure in Line~\ref{line:alg-output}, to bound the cost of the different solutions we obtain.
The variable $\yS$ represents the coloring duration of a set $S \subseteq \V$ in \SSW{}.
We use $\SSW$ and $\Rec$ to refer to the solutions obtained using the call to the 3-approximation algorithm and the recursive call, respectively.

We begin by defining some notations and classifying sets in $\Dalg$.
Let $\Fopt$ denote the forest used in the optimal solution. 
Similarly, let $\Dopt$ denote the family for which the optimal solution pays the penalty, i.e., $\Dopt = \pen(\Fopt)$.

\begin{definition}\label{def:paid}
    Recall that for instance $\I$, $\Dalg$ is the family returned by $\SSW$ in the output (Line \ref{line:alg-output} of Algorithm \ref{alg:pcf2}). We partition $\Dalg$ into two disjoint families as follows:
    \begin{align*}
        \Dpaid &= \{ S \in \Dalg \mid \pi(\{S\} \mid \Dopt) = 0 \} \\
        \Dunpaid &= \{ S \in \Dalg \mid \pi(\{S\} \mid \Dopt) > 0 \}
    \end{align*}
\end{definition}

From this definition, the following property follows for $\Dunpaid$.

\begin{lemma}
\label{lm:unpaid-means-notin-dopt}
    If a set $S \in \Dunpaid$, then $S \notin \Dopt$.
\end{lemma}
\begin{proof}
    We prove this by contradiction. Assuming $S \in \Dopt$, we have:
    \begin{align*}
        \pi(\{S\} \mid \Dopt) &= \pi(\{S\} \cup \Dopt) - \pi(\Dopt) \\
        &= \pi(\Dopt) - \pi(\Dopt) = 0. \tag{$S \in \Dopt$}
    \end{align*}
    On the other hand, since $S \in \Dunpaid$, from Definition \ref{def:paid} we have:
    \begin{align*}
        \pi(\{S\} \mid \Dopt) > 0 \text{,}
    \end{align*}
    which is a contradiction.
\end{proof}

In the following lemma, we show that all sets in $\Dunpaid$ cut at least one edge of the forest in the optimal solution. 

\begin{lemma}
\label{lm:unpaid-eq-cut-opt}
    For any set $S \in \Dunpaid$, there is at least one edge $e$ in $\Fopt$ such that $e \in \delta(S)$.
\end{lemma}
\begin{proof}
    From lemma \ref{lm:unpaid-means-notin-dopt}, we know that $S \notin \Dopt = \pen(\Fopt)$. Therefore, by Definition \ref{def:pen}
    \begin{align*}
        \lvert\delta(S) \cap \Fopt\rvert > 0 \text{,}
    \end{align*}
    which means $\exists e \in \delta(S) \cap \Fopt$.
\end{proof}

Now, since all sets in $\Dunpaid$ cut at least one edge of $\Fopt$, we can divide this family into two families as follows:

\begin{definition}
\label{def:single-mult}
    Let us partition sets in family $\Dunpaid$ into two disjoint families $\Dsingle$ and $\Dmult$ based on the number of edges of the forest $\Fopt$ they cut. Formally,
    \begin{align*}
        \Dsingle &= \{ S \in \Dunpaid \mid |\delta(S) \cap \Fopt| = 1 \} \\
        \Dmult &= \{ S \in \Dunpaid \mid |\delta(S) \cap \Fopt| > 1 \}.
    \end{align*}
    In addition, we define $\ysingle$ and $\ymult$ based on the $\yS$ values for sets in $\Dsingle$ and $\Dmult$:
    \begin{align*}
        \ysingle &= \sum_{S \in \Dsingle} \yS &
        \ymult &= \sum_{S \in \Dmult} \yS
    \end{align*}    
\end{definition}

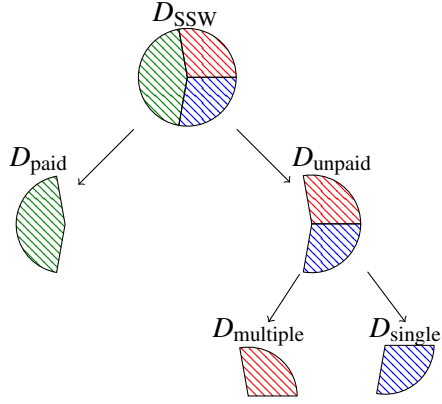
\begin{figure}[ht]
    \centering
    \begin{tikzpicture}[scale=1]

\def\R{Red}
\def\G{Green}
\def\B{Blue}
\def\d{0.65}

\tikzstyle{node} = [circle, minimum size=2*\d, inner sep=0pt]
\tikzstyle{smallnode} = [circle, minimum size=\d, inner sep=0pt]
\tikzstyle{pie} = [pattern=north west lines]

\node[node] (Dssw) at (0, 6*\d) {\begin{tikzpicture}
    \draw[pie, pattern color=\R] (0,0) -- (0:\d) arc(0:100:\d) -- cycle;
    \draw[pie, pattern color=\G] (0,0) -- (100:\d) arc(100:260:\d) -- cycle;
    \draw[pie, pattern color=\B] (0,0) -- (260:\d) arc(260:360:\d) -- cycle;
\end{tikzpicture}};

\node[node] (Dpaid) at (-3*\d, 3*\d) {\begin{tikzpicture}
    \draw[pie, pattern color=\G] (0,0) -- (100:\d) arc(100:260:\d) -- cycle;
\end{tikzpicture}};

\node[node] (Dunpaid) at (3*\d, 3*\d) {\begin{tikzpicture}
    \draw[pie, pattern color=\R] (0,0) -- (0:\d) arc(0:100:\d) -- cycle;
    \draw[pie, pattern color=\B] (0,0) -- (260:\d) arc(260:360:\d) -- cycle;
\end{tikzpicture}};

\node[smallnode] (Dmultiple) at (1.7*\d, 0) {\begin{tikzpicture}
    \draw[pie, pattern color=\R] (0,0) -- (0:\d) arc(0:100:\d) -- cycle;
\end{tikzpicture}};

\node[smallnode] (Dsingle) at (4.5*\d, 0) {\begin{tikzpicture}
    \draw[pie, pattern color=\B] (0,0) -- (260:\d) arc(260:360:\d) -- cycle;
\end{tikzpicture}};

\draw[->] (Dssw) -- (Dpaid);
\draw[->] (Dssw) -- (Dunpaid);
\draw[->] (Dunpaid) -- ($(Dmultiple) + (0, 1*\d)$);
\draw[->] (Dunpaid) -- ($(Dsingle) + (0, 1*\d)$);

\node at (-3*\d, 4.3*\d) {$D_{\text{paid}}$};
\node at (3*\d, 4.3*\d) {$D_{\text{unpaid}}$};
\node at (1.5*\d, 0.8*\d) {$D_{\text{multiple}}$};
\node at (4.5*\d, 0.8*\d) {$D_{\text{single}}$};
\node at (0*\d, 7.3*\d) {$D_{\text{SSW}}$};

\end{tikzpicture}
    \caption{Tight sets returned by \SSW{}, denoted by $\Dalg$, are partitioned into $\Dpaid$ and $\Dunpaid$, with $\Dunpaid$ further divided into $\Dsingle$ and $\Dmult$.}
    \label{fig:partitions-def}
\end{figure}

Up to this point, our approach involves deciding to pay the penalty of $\Dalg$, which we then divide into sets based on whether paying their penalty increases the penalty of the family that the optimal solution pays. We have shown that sets increasing the optimal solution's penalty must cut at least one edge of the optimal solution. Thus, we classify these sets based on whether they cut exactly one edge or more than one edge.
Figure~\ref{fig:partitions-def} demonstrates the classification of sets in $\Dalg$.

Next, we use this classification to establish a lower bound for the cost of the optimal solution and upper bounds for $\SSW$ and $\Rec$ solutions. We will employ induction to complete the proof of the approximation guarantee. We begin by providing a lower bound for the cost of the optimal solution.

\begin{lemma}
\label{lm:optcost}
    For instance $I$, the cost of the optimal solution can be bounded as follows:
    $$\cost(\Opt) \ge \sum_{S \subseteq V}\yS + \ymult.$$
\end{lemma}
\begin{proof}
    First, we we know that
    \begin{align*}
        \cost(\Opt) = \ce(\Fopt) + \pi(\Dopt).
    \end{align*}
    We obtain lower bounds for both the cost of $\Fopt$ and the penalty paid by the optimal solution for $\Dopt$ in terms of $\yS$ values.
    Based on Corollary \ref{corollary:family_constraint}, for any family $\mathcal{S}$ we have $\sum_{S \in \mathcal{S}} \yS \le \pi(\mathcal{S})$. Applying this to $\Dopt$ results in the following bound:
    \begin{align*}
         \pi(\Dopt) \ge \sum_{S \in \Dopt} \yS. \tag{Corollary \ref{corollary:family_constraint}}
    \end{align*}
    Also, based on Corollary \ref{corollary:edge_constraint}, for any $e \in \E$, we have $\sum_{S: e \in \deltaS} \yS \le \ce(e)$. Therefore for $\ce(\Fopt)$ we have
    \begin{align*}
        \ce(\Fopt) &= \sum_{e \in \Fopt} \ce(e) \\
        &\ge \sum_{e \in \Fopt} \sum_{\substack{S \subseteq \V\\e \in \deltaS}} \yS \tag{Corollary \ref{corollary:edge_constraint}}\\
        &= \sum_{S \subseteq \V} \yS \cdot |\delta(S) \cap \Fopt| \tag{Reordering summation}\\
        &= \sum_{S \notin \Dopt} \yS \cdot |\delta(S) \cap \Fopt| \tag{$\forall S \in \Dopt: |\delta(S) \cap \Fopt| = 0$}\\
        &\ge \sum_{S \notin \Dopt} \yS + \sum_{\substack{S \notin \Dopt\\|\delta(S) \cap \Fopt| \ge 2}} \yS \tag{$\forall S \notin \Dopt: |\delta(S) \cap \Fopt| \ge 1$}\\
        &\ge \sum_{S \notin \Dopt} \yS + \sum_{\substack{S \in \Dunpaid\\|\delta(S) \cap \Fopt| \ge 2}} \yS \tag{Lemma \ref{lm:unpaid-means-notin-dopt}} \\
        &= \sum_{S \notin \Dopt} \yS + \sum_{S \in \Dmult} \yS \tag{Definition \ref{def:single-mult}} \\
        &= \sum_{S \notin \Dopt} \yS + \ymult.
    \end{align*}
    Combining these components completes the proof:
    \begin{align*}
        \cost(\Opt) &= \ce(\Fopt) + \pi(\Dopt)\\
        &\ge \sum_{S \notin \Dopt} \yS + \ymult + \sum_{S \in \Dopt} \yS\\
        &= \sum_{S \subseteq \V} \yS + \ymult.
    \end{align*}
\end{proof}

Next, we give upper bounds for the costs of solutions $\SSW$ and $\Rec$ which are used in our algorithm. We begin with an upper bound for $\SSW$ which uses the lower bound obtained for $\Opt$.

\begin{lemma} \label{lm:algcost}
    For instance $I$, we have the following bound for the cost of the solution $\Alg$:
    $$\cost(\Alg) \le 2\cost(\Opt) - 2\ymult + \pi(\Dalg).$$
\end{lemma}
\begin{proof}
    The cost of $\SSW$ can be upper bounded by 
    \begin{align*}
        \cost(\Alg) &= \ce(\Falg) + \pi(\pen(\Falg))\\
        &\le \ce(\Falg) + \pi(\Dalg). \tag{Lemma \ref{lm:ssw_penalty_le_Dm_penalty}}
    \end{align*}
    In addition, based on Lemma \ref{lm:ssw_forest_cost} in the analysis of the 3-approximation algorithm, we have:
    \begin{align*}
        \ce(\Falg) \le 2\sum_{S \subseteq V} \yS.
    \end{align*}
    We can complete the proof using the inequality from Lemma \ref{lm:optcost}:
    \begin{align*}
        \cost(\Alg) &\le \ce(\Falg) + \pi(\Dalg)\\
        &\le 2\sum_{S \subseteq V} \yS + \pi(\Dalg) \\
        &= 2\sum_{S \subseteq V} \yS + 2\ymult- 2\ymult + \pi(\Dalg) \\
        &\le 2\cost(\Opt) - 2\ymult + \pi(\Dalg) \tag{Lemma \ref{lm:optcost}}
    \end{align*}
\end{proof}

To analyze the cost of the solution $\Rec$, we first introduce the solution $\Solr$ below to establish an upper bound on the cost of the optimal solution for instance $\R$. 
Assuming by induction that our algorithm achieves a 2-approximate solution for instance $\R$, we then use this upper bound to bound the cost of $\Rec$ on the original instance $\I$.

\begin{definition}
Consider the forest $\Fopt$ selected by the optimal solution $\Opt$. 
For any set $S \in \Dsingle$, remove the unique edge in $\Fopt$ cut by $S$. 
Let $\Fr$ be the resulting forest and $\Dr = \pen(\Fr)$ denote the family for which the penalty needs to be paid for forest $\Fr$. 
We define $\Solr$ as the solution represented by the forest $\Fr$, which pays the penalty for family $\Dr$.
\end{definition}

We analyze the cost of $\Solr$ in the next two lemmas. 
First, we introduce a lower bound for the cost of the removed edges to analyze how the cost of the forest $\Fr$ compares to $\Fopt$.
Then, we show that removing these edges from $\Fopt$ does not affect the penalty that needs to be paid when considering the penalty function for instance $\R$. 
Note that the inequality in Lemma~\ref{lm:rforestpen} is in fact an equality, as $\Dopt \subseteq \closure(\Dr)$. 
However, proving the inequality is sufficient for our analysis.

\begin{lemma}\label{lm:rforestcost}
    The weight of the forest $\Fr$ is at most $\ce(\Fopt) - \ysingle$.
\end{lemma}
\begin{proof}
Forest $\Fr$ is obtained from $\Fopt$ by removing all edges that are cut by sets in 
$\Dsingle$. 
Let $\Er$ denote this set of edges.
We have
\begin{align*}
    \ce(\Er) &= \sum_{e\in\Er} \ce_e\\
    &\geq\sum_{e\in\Er} \sum_{\substack{S\subseteq\V\\ e\in\deltaS}} \yS \tag{Corollary \ref{corollary:edge_constraint}}\\
    &=\sum_{S\subseteq\V} \sum_{\substack{e\in\Er\\e\in\deltaS}} \yS \tag{Reordering terms}\\
    &=\sum_{S\subseteq\V} \yS\cdot\lvert \delta(S)\cap\Er\rvert  \\
    &\geq\sum_{S\in\Dsingle} \yS\cdot\lvert \delta(S)\cap\Er\rvert \tag{Restricting to $\Dsingle$}\\
    &=\sum_{S\in\Dsingle} \yS \tag{Definition of $\Er$}\\
    &=\ysingle. \tag{Definition \ref{def:single-mult}}
\end{align*}
Therefore $\ce(\Fr)=\ce(\Fopt) - \ce(\Er) \leq \ce(\Fopt) - \ysingle$.
\end{proof}

\begin{lemma} \label{lm:rforestpen}
    For instance $\R$, the penalty paid by $\Solr$ is no more than the penalty paid by $\Opt$:
    $$\pir(\Dr) \le \pir(\Dopt).$$
\end{lemma}
\begin{proof}
Consider an arbitrary order for removing the edges $\Er$ to transform $\Fopt$ into $\Fr$. For $0\leq i \leq \lvert \Er \rvert$, let $\Fr^i$ be the forest after removing the first $i$ edges in this order with $\Fr^0=\Fopt$ and $\Fr^{\lvert \Er \rvert}=\Fr$ and $\Dr^i=\pen(\Fr^i)$. We will prove the desired statement by showing that for each $0\leq i <\lvert \Er \rvert$, $\pir(\Dr^{i+1})\leq\pir(\Dr^i)$.

First, note that since $\pir(\D)=\pi(\D\mid\Dalg)=\pi(\D\cup\Dalg) - \pi(\Dalg)$, to demonstrate that $\pir(\Dr^{i+1})\leq\pir(\Dr^i)$, it suffices to prove the equivalent condition $\pi(\Dr^{i+1}\cup\Dalg) \leq \pi(\Dr^i\cup\Dalg)$.
Now, assume we establish that $\Dr^{i+1}\subseteq \closure(\Dr^i\cup\Dalg)$.
Then, by Lemma \ref{lm:closure} and the monotonicity of $\pi$, we have 
\begin{align*}
\pi(\Dr^{i+1}\cup\Dalg) &\leq \pi(\Dr^{i+1}\cup\closure(\Dr^i\cup\Dalg)) \tag{Monotonicity of $\pi$}\\
&\leq\pi(\closure(\Dr^i\cup\Dalg)) \tag{Assuming $\Dr^{i+1}\subseteq \closure(\Dr^i\cup\Dalg)$}\\
&=\pi(\Dr^i\cup\Dalg) \tag{Lemma \ref{lm:closure}},
\end{align*}
which is our desired statement. Therefore, we focus on showing that $\Dr^{i+1}\subseteq \closure(\Dr^i\cup\Dalg)$.

Note that $\Dr^{i+1}$ is the $\closure$ of the family of the components of $\Fr^{i+1}$ by Lemma \ref{lm:compcl}. 
Thus, it suffices to show that each component of $\Fr^{i+1}$ is contained in $\closure(\Dr^i\cup\Dalg)$, i.e., $\cc(\Fr^{i+1}) \subseteq \closure(\Dr^i\cup\Dalg)$, which concludes that 
$$\Dr^{i+1}=\closure(\cc(\Fr^{i+1})) \subseteq \closure(\closure(\Dr^i\cup\Dalg)) = \closure(\Dr^i\cup\Dalg)\text{.}$$
The forest $\Fr^{i+1}$ is obtained from the forest $\Fr^{i}$ by removing an edge $e$, which splits some component $C$ containing this edge into two components $C_1$ and $C_2$. 
Firstly, note that for any other component $C'$ of $\Fr^{i+1}$ except $C_1$ and $C_2$, $C'$ is also a component of $\Fr^i$ and therefore contained in $\Dr^i$. 

Now, consider a set $S\in\Dsingle$ that cuts $e$ and take the intersection of this set with $C$. 
Since $S$ cuts $e$, it must contain exactly one of its endpoints. Assume that $S$ includes the endpoint of $e$ in $C_1$. Then, $S$ must include every vertex in $C_1$; otherwise, as $\Fopt$ restricted to $C_1$ is connected, there would be a second edge of $\Fopt$ cut by $S$, which contradicts $S$ being in $\Dsingle$. 
Similarly, no vertex in $C_2$ can be included in $S$. 
Therefore, $S\cap C=C_1$ and $C\setminus S=C_2$.
This is illustrated in Figure \ref{fig:remain-closure}.

Finally, since $S\in\Dalg$ and $C\in\Dr^i$, the components $C_1 = C \cap S$ and $C_2 = C\setminus S$ are included in $\closure(\Dr^i\cup\Dalg)$ due to this family being closed under intersection and set difference operations. This shows that $\Dr^{i+1}\subseteq \closure(\Dr^i\cup\Dalg)$, which completes the proof.
\end{proof}

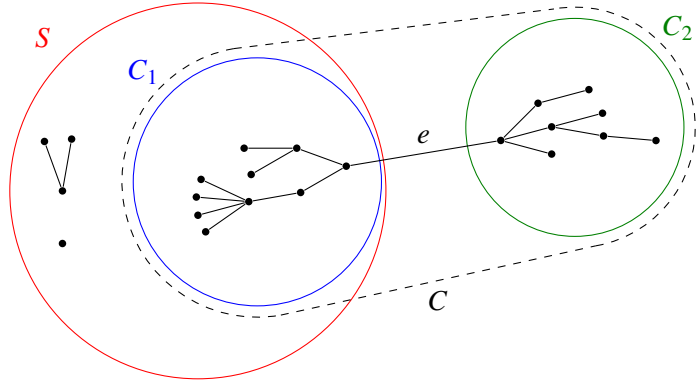
\begin{figure}[t]
    \centering
\begin{tikzpicture}[scale=1]

\def\gc{Black}
\def\es{0.7}

\node[Red, draw, circle, minimum size=5cm, label=above left:{\color{Red} $S$}] (S) at (0,0) {};
\node (S2) at (6,1) {};

\node[fill=\gc, circle, inner sep=1pt, label=85:] (v) at ($(S)!2cm!(S2)$){};
\node[fill=\gc, circle, inner sep=1pt, label=135:] (u) at ($(S2)!2cm!(S)$){};

\node[Blue, draw, circle, minimum size=3.3cm, label=above left:{\color{Blue} $C_1$}] (SP1) at ($(v) + (10:-1.2cm)$) {};
\node[Green, draw, circle, minimum size=2.9cm, label=above right:{\color{Green} $C_2$}] (SP2) at ($(u) + (10:1cm)$) {};

\draw[Black, dashed] ($(SP1) + (100: 1.8cm)$) arc[start angle=100, end angle=280, radius=1.8cm];
\draw[Black, dashed] ($(SP2) + (-80: 1.6cm)$) arc[start angle=-80, end angle=100, radius=1.6cm];
\draw[Black, dashed] ($(SP1) + (100: 1.8cm)$) -- ($(SP2) + (100: 1.6cm)$);
\draw[Black, dashed] ($(SP1) + (280: 1.8cm)$) -- node[below] {$C$} ($(SP2) + (-80: 1.6cm)$);

\node[fill=\gc, circle, inner sep=1pt] (a1) at ($(v) + (30: -\es)$) {};
\node[fill=\gc, circle, inner sep=1pt] (a2) at ($(v) + (-20: -\es)$) {};
\node[fill=\gc, circle, inner sep=1pt] (a3) at ($(a2) + (0: -\es)$) {};
\node[fill=\gc, circle, inner sep=1pt] (a4) at ($(a2) + (30: -\es)$) {};
\node[fill=\gc, circle, inner sep=1pt] (a5) at ($(a1) + (10: -\es)$) {};
\node[fill=\gc, circle, inner sep=1pt] (a6) at ($(a5) + (-25: -\es)$) {};
\node[fill=\gc, circle, inner sep=1pt] (a7) at ($(a5) + (-5: -\es)$) {};
\node[fill=\gc, circle, inner sep=1pt] (a8) at ($(a5) + (15: -\es)$) {};
\node[fill=\gc, circle, inner sep=1pt] (a9) at ($(a5) + (35: -\es)$) {};

\draw[\gc] (v) -- (a1);
\draw[\gc] (v) -- (a2);
\draw[\gc] (a2) -- (a3);
\draw[\gc] (a2) -- (a4);
\draw[\gc] (a1) -- (a5);
\draw[\gc] (a5) -- (a6);
\draw[\gc] (a5) -- (a7);
\draw[\gc] (a5) -- (a8);
\draw[\gc] (a5) -- (a9);

\node[fill=\gc, circle, inner sep=1pt] (b1) at ($(u) + (45: \es)$) {};
\node[fill=\gc, circle, inner sep=1pt] (b2) at ($(u) + (15: \es)$) {};
\node[fill=\gc, circle, inner sep=1pt] (b3) at ($(u) + (-15: \es)$) {};
\node[fill=\gc, circle, inner sep=1pt] (b4) at ($(b1) + (15: \es)$) {};
\node[fill=\gc, circle, inner sep=1pt] (b5) at ($(b2) + (15: \es)$) {};
\node[fill=\gc, circle, inner sep=1pt] (b6) at ($(b2) + (-10: \es)$) {};
\node[fill=\gc, circle, inner sep=1pt] (b7) at ($(b6) + (-5: \es)$) {};

\draw[\gc] (u) -- (b1);
\draw[\gc] (u) -- (b2);
\draw[\gc] (u) -- (b3);
\draw[\gc] (b1) -- (b4);
\draw[\gc] (b2) -- (b5);
\draw[\gc] (b2) -- (b6);
\draw[\gc] (b6) -- (b7);

\node[fill=\gc, circle, inner sep=1pt] (c1) at (-1.8, 0) {};
\node[fill=\gc, circle, inner sep=1pt] (c2) at ($(c1) + (110: \es)$) {};
\node[fill=\gc, circle, inner sep=1pt] (c3) at ($(c1) + (80: \es)$) {};
\node[fill=\gc, circle, inner sep=1pt] (c4) at ($(c1) + (-90: \es)$) {};

\draw[\gc] (c1) -- (c2);
\draw[\gc] (c1) -- (c3);

\draw[\gc] (v) -- (u) node[black, midway, above] {$e$};
\end{tikzpicture}
    \caption{An illustration of the steps in Lemma \ref{lm:rforestpen}. When removing the sole edge $e$ cut by set $\redS$, component $\blackC$ is split into components $\blueCone$ and $\greenCtwo$. It can be seen that any vertex in $\greenCtwo$ being in $\redS$ or any vertex in $\blueCone$ being outside $\redS$ will lead to another edge cut by $\redS$. Therefore, $\redS \cap \blackC$ will be $\blueCone$, as shown here.}
    \label{fig:remain-closure}
\end{figure}

Now, we use the solution $\Solr$ to bound the optimal solution of instance $\R$, which costs less than $\Solr$ due to its optimality.
Note that we use $\costr$ to denote the cost of a solution for instance $\R$.

\begin{lemma} \label{lm:costrr}
    Let $\Optr$ be an optimal solution for instance $\R$. Then, we have 
    $$\costr(\Optr) \le \cost(\Opt) - (\pi(\Dopt) - \pi(\Dopt\mid\Dalg)) - \ysingle.$$
\end{lemma}
\begin{proof}
Consider the solution $\Solr$ along with its associated forest $\Fr$ and family $\Dr$. We can use this solution to upper bound the cost of $\Optr$ as follows:
\begin{align*}
    \costr(\Optr) &\leq \costr(\Solr) \tag{$\Optr$ is optimal for instance $\R$}\\&= \ce(\Fr) + \pir(\Dr)\\
    &\leq\ce(\Fr) + \pir(\Dopt) \tag{Lemma \ref{lm:rforestpen}}\\
    &=\ce(\Fr) + \pi(\Dopt\mid\Dalg) \tag{Definition of $\pir$}\\
    &\leq\ce(\Fopt) - \ysingle + \pi(\Dopt\mid\Dalg) \tag{Lemma \ref{lm:rforestcost}}\\
    &=(\cost(\Opt) - \pi(\Dopt)) - \ysingle + \pi(\Dopt\mid\Dalg) \tag{$\cost(\Opt)=\ce(\Fopt)+\pi(\Dopt)$}\\
    &=\cost(\Opt) - (\pi(\Dopt) - \pi(\Dopt\mid\Dalg)) - \ysingle.
\end{align*}
\end{proof}

We use the following lemma to bound the cost of the solution $\Rec$ on instance $\I$ based on its cost on instance $\R$.

\begin{lemma}\label{lm:costrbound}
    For any solution $\Sol$, 
    $$
        \cost(\Sol) \leq \costr(\Sol) + \pi(\Dalg).
    $$
\end{lemma}
\begin{proof}
Using the definition of $\costr$, we can write
\begin{align*}
    \costr(\Sol) &= \ce(\F_\Sol) + \pir(\D_\Sol)\\
    &=\ce(\F_\Sol) + \pi(\D_\Sol\mid\Dalg)\\
    &=\ce(\F_\Sol) + \pi(\D_\Sol\cup\Dalg)-\pi(\Dalg) \tag{Definition \ref{def:marg}}\\
    &\geq\ce(\F_\Sol) + \pi(\D_\Sol)-\pi(\Dalg) \tag{Monotonicity of $\pi$}\\
    &=\cost(\Sol)-\pi(\Dalg).
\end{align*}
\end{proof}

Next, we relate $\ysingle$ and $\ymult$ so we can give an upper bound for $\Rec$ in terms of $\ymult$, which is consistent with our upper bound for $\Alg$.

\begin{lemma} \label{lm:b1pb2}
    We have 
    $$\ysingle + \ymult \geq \pi(\Dalg\mid\Dopt).$$
\end{lemma}
\begin{proof}
    By Definition \ref{def:single-mult}, $\ysingle + \ymult = \sum_{S\in\Dunpaid} \yS$, so we are looking to show that $$\sum_{S\in\Dunpaid} \yS \geq \pi(\Dalg\mid\Dopt).$$
    
    We know by Definition \ref{def:paid} that for each set $S \in \Dpaid$, $\pi(\{S\}\mid\Dopt)=0$. Then, the submodularity of $\pi(\ \cdot\mid\Dopt)$ along with Lemma \ref{lm:subadd} implies that $\pi(\Dpaid\mid\Dopt)=0$. This means that $$\pi(\Dpaid\cup\Dopt)=\pi(\Dopt).$$
    So, we can conclude that
    \begin{align*}
        \pi(\Dalg\mid\Dopt) &= \pi(\Dalg\cup\Dopt) - \pi(\Dopt)\\
        &=\pi(\Dalg\cup\Dopt) - \pi(\Dpaid\cup\Dopt)\\
        &=\pi(\Dunpaid\cup\Dpaid\cup\Dopt) - \pi(\Dpaid\cup\Dopt)\tag{Definition~\ref{def:paid}}\\
        &=\pi(\Dunpaid\mid\Dpaid\cup\Dopt)\\
        &\leq\pi(\Dunpaid\mid\Dpaid) \tag{Lemma \ref{lm:submod}}\\
        &= \pi(\Dunpaid\cup\Dpaid) - \pi(\Dpaid)\\
        &= \pi(\Dalg) - \pi(\Dpaid) \tag{Definition~\ref{def:paid}}\\
        &\leq \pi(\Dalg) - \sum_{S\in\Dpaid}\yS \tag{Corollary \ref{corollary:family_constraint}}\\
        &= \sum_{S\in\Dalg}\yS - \sum_{S\in\Dpaid}\yS \tag{Lemma \ref{lm:tight}}\\
        &= \sum_{S\in\Dunpaid} \yS\text{.} \tag{Definition~\ref{def:paid}}
    \end{align*}
\end{proof}

Finally, we use induction on the recursive steps of the algorithm to prove that \PCF{} returns a 2-approximate solution. The base case of our induction is when \PCF{} does not recursively call itself. 
In the induction step, we show that if the recursive call returns a 2-approximate solution for the new instance $\R$, then the minimum solution between $\SSW$ and $\Rec$ is a 2-approximate solution for instance $\I$. 
The following lemma is crucial for the induction step as it helps bound the cost of $\Rec$ using the induction hypothesis. 
After establishing this lemma, we proceed to prove the approximation guarantee in Theorem~\ref{theorem:main_theorem}.

\begin{lemma} \label{lm:rbound}
    If a solution $\Sol$ is a 2-approximate solution for instance $\R$, we have
    $$\cost(\Sol) \le 2\cost(\Opt) + 2\ymult - \pi(\Dalg)$$
    for instance $\I$.
\end{lemma}
\begin{proof}
First, we note that by Lemma \ref{lm:b1pb2} $$\ysingle \geq \pi(\Dalg|\Dopt) - \ymult.$$
Now, we can bound the cost of the solution $\Sol$ in instance $\R$ using the cost of the optimal solution in instance $\I$.
\begin{align*}
    \costr(\Sol) \leq{}& 2\costr(\Optr) \tag{Lemma's assumption}\\ \leq{}&2(\cost(\Opt) - (\pi(\Dopt) - \pi(\Dopt\mid\Dalg)) - \ysingle)\tag{Lemma \ref{lm:costrr}}\\
    \leq{}&2(\cost(\Opt) - \pi(\Dopt) + \pi(\Dopt\mid\Dalg) - \pi(\Dalg|\Dopt) + \ymult)\tag{Lemma \ref{lm:b1pb2}}\\
    ={}&2(\cost(\Opt) - \pi(\Dopt) + (\pi(\Dopt\cup\Dalg) - \pi(\Dalg))\\& - (\pi(\Dalg\cup\Dopt) - \pi(\Dopt)) + \ymult)\\
    ={}&2(\cost(\Opt) - \pi(\Dalg) + \ymult) \text{.}
\end{align*}
Finally, we can bound the cost of the solution $\Sol$ in instance $\I$. 
\begin{align*}
    \cost(\Sol) &\leq \costr(\Sol) + \pi(\Dalg) \tag{Lemma~\ref{lm:costrbound}}\\
    &\leq 2(\cost(\Opt) - \pi(\Dalg) + \ymult) + \pi(\Dalg)\\
    &= 2\cost(\Opt) + 2\ymult - \pi(\Dalg).
\end{align*}
\end{proof}

Finally, we combine the upper bounds for the two solutions $\SSW$ and $\Rec$ to establish the approximation ratio of the solution returned by Algorithm \ref{alg:pcf2}. We note that we assume the algorithm terminates in the proof of Theorem \ref{theorem:main_theorem}, before formally proving this in Theorem \ref{theorem:main_time_complexity}.
\begin{theorem}
\label{theorem:main_theorem}
    The solution $\Sol$ returned by the algorithm \PCF{} has a cost of at most $2\cost(\Opt)$.
\end{theorem}
\begin{proof}
    We prove this by induction on the recursive tree of the algorithm.
    In the base case, which is the last call to the \PCF{} the algorithm returns the solution $\SSW$ without making a recursive call. In this case, $\pi(\Dalg)=0$, so the algorithm does not pay any penalties. Then, by Lemma \ref{lm:algcost}, we have
    \begin{align*}
    \cost(\Sol) =\cost(\Alg) \leq 2\cost(\Opt) - 2 \ymult  + \pi(\Dalg) \leq  2\cost(\Opt).
    \end{align*}

    Now, consider the case where we also obtain a recursive solution $\Rec$. In this case, by Lemma \ref{lm:algcost}, we have the following bound for the non-recursive solution $\Alg$
    $$
    \cost(\Alg) \leq 2\cost(\Opt) - 2 \ymult  + \pi(\Dalg).
    $$
    By induction, the recursive solution $\Rec$ is a 2-approximation for instance $\R$. So, by Lemma \ref{lm:rbound} we have
    $$
    \cost(\Rec) \leq 2\cost(\Opt) + 2 \ymult  -\pi(\Dalg).
    $$
    Combining the two bounds, we can see that
    \begin{align*}
    \cost(\Sol)&=\min(\cost(\Alg),\cost(\Rec))\\
    & \leq \frac{\cost(\Alg)+\cost(\Rec)}{2}\\
    &\leq \frac{2\cost(\Opt) - 2 \ymult  + \pi(\Dalg) + 2\cost(\Opt) + 2 \ymult  - \pi(\Dalg)}{2}\\
    &\leq 2\cost(\Opt).
    \end{align*}
\end{proof}
Lastly, we show that our algorithm operates in polynomial time.

\begin{theorem} \label{theorem:main_time_complexity}
    Algorithm \ref{alg:pcf2} always terminates and runs in polynomial time.    
\end{theorem}
\begin{proof}
    First, we note that each recursive step runs in polynomial time. The runtime for each step is dominated by the call to the 3-approximation algorithm $\SSW$, which runs in polynomial time by Lemma \ref{lm:ssw_poly_time}. 
    It's important to note that this runtime depends on having polynomial time access to an oracle for the input penalty function $\pi'$ . 
    This access is guaranteed for the initial penalty function $\pi$, but the penalty function used changes throughout our algorithm. 
    However, based on Lemma~\ref{lm:union_marginal_cost}, at any step of the algorithm, the penalty function $\pi'(\D)$ can be thought of as $\pi(\D\mid\Dbase)$, where $\pi$ is the penalty function from the initial call and $\Dbase$ is the union of all the tight families $\Dalg$ used to define $\pir$ functions so far. 
    This value can be computed using two calls to the oracle for the original penalty function $\pi$, so, assuming the size of $\Dbase$ remains polynomial, the runtime of each step will be polynomial too. 
    Next, we will show that the recursion depth of the algorithm is polynomial. This, along with the fact that in each recursion the number of elements added to $\Dbase$ is polynomial, guarantees that $\Dbase$ will have a polynomial size and polynomial-time access to the penalty function at each step is maintained.

        \begin{figure}[t]
        \centering
        \begin{tikzpicture}[scale=1.2]

\def\B{Blue}
\def\R{Green}
\def\G{Red}
\def\xmargin{0.7}
\def\ymargin{0.2}
\def\rad{1}

\def\d{60}
\def\dd{7}
\def\ddd{4}
\def\sr{\rad - 0.1}
\def\sbr{\sr + 0.1}

\def\srad{0.5}
\def\ssrad{\srad + 0.1}
\def\nd{73}
\def\ndp{32.7}
\def\ndd{11}
\def\nddd{2.3}
\def\nnd{22}
\def\dddd{2}
\def\ndddd{3}

\foreach \x in {0, 1}
\foreach \y in {0, 1, 2} {
    \coordinate (center\x\y) at ($(\x * \rad * 2 + \x * \xmargin, \y * \rad * 2 + \y * \ymargin)$);
    
    \coordinate (top1r\x\y) at ($(center\x\y) + (\d:\sr)$);

    \coordinate (top1l\x\y) at ($(center\x\y) + (\d+\dd:\sr)$); 

    \coordinate (bot1r\x\y) at ($(center\x\y) + (-\d:\sr)$);

    \coordinate (bot1l\x\y) at ($(center\x\y) + (-\d-\dd:\sr)$); 

    \coordinate (top3r\x\y) at ($(center\x\y) + (\nd:\srad)$);
    
    \coordinate (top3l\x\y) at ($(center\x\y) + (\nd+\ndd:\srad)$);

    \coordinate (top2r\x\y) at ($(center\x\y) + (\nd - \ndddd:\ssrad)$);
    
    \coordinate (top2l\x\y) at ($(center\x\y) + (\nd+\ndd - \ndddd:\ssrad)$);
}

\foreach \x/\y/\c in {0/0/\B, 0/1/\B, 0/2/\B, 1/2/\R} {
    \draw[\c] (center\x\y) circle (\rad);
};

\foreach \x/\y/\c in {0/0/\B, 0/1/\B, 1/1/\R, 1/0/\R} {
    \draw[\c] (top1r\x\y) arc[start angle=\d, end angle=-\d, radius=\sr];
};
\foreach \x/\y/\c in {0/1/\B,  1/1/\R} {
    \draw[\c] (top1r\x\y) arc[start angle=180 - \d, end angle=180 + \d, radius=\sr];
};

\def\f{61.5}
\def\fr{\sr+0.4}
\foreach \x/\y/\c in {0/1/\G} {
    \coordinate (top1) at ($(center\x\y) + (\f:\sr+0.05)$);
    \draw[\c, very thick] (top1) arc[start angle=180 - \f, end angle=180 + \f, radius=\sr+0.05];
    \draw[\c, very thick] (top1) arc[start angle=\f, end angle=-\f, radius=\sr + 0.05];
};

\foreach \x/\y/\c in {0/1/\B, 1/1/\R} {
    \draw[\c] (top1l\x\y) arc[start angle=180 - \d + \ddd, end angle=180 + \d - \ddd, radius=\sbr];
}
\foreach \x/\y/\c in {0/0/\B, 0/1/\B, 1/1/\R, 1/0/\R} {
    \draw[\c] (top1l\x\y) arc[start angle=\d + \dd, end angle=360 -\d - \dd, radius=\sr];
};

\foreach \x/\y/\c in {0/0/\G} {
    \draw[\c, very thick] (center\x\y) circle (\srad + 0.05);
}

\foreach \x/\y/\c in {0/0/\B, 1/0/\R} {
    
    \draw[\c] (top1r\x\y) arc[start angle=180 - \d, end angle=180 - \d + \nnd - \dddd, radius=\sr];

    \draw[\c] (bot1r\x\y) arc[start angle=180 + \d, end angle=180 + \d - \nnd + \dddd, radius=\sr];

    \draw[\c] (top1l\x\y) arc[start angle=180 - \d + \ddd, end angle=180 - \d + \ddd + \nnd - \dddd, radius=\sbr];

    \draw[\c] (bot1l\x\y) arc[start angle=180 + \d - \ddd, end angle=180 + \d - \ddd - \nnd + \dddd, radius=\sbr];

    \draw[\c] (top2l\x\y) arc[start angle=\nd + \ndd - \ndddd, end angle=360 -\nd - \ndd + \ndddd, radius=\ssrad];

    \draw[\c] (top2r\x\y) arc[start angle=\nd-\nddd, end angle=-\nd+\nddd, radius=\ssrad];
}

\foreach \x/\y/\c in {0/0/\B, 1/0/\R} {
    \draw[\c] (top3r\x\y) arc[start angle=180 - \ndp, end angle=180 + \ndp-1, radius=\sr];

    \draw[\c] (top3l\x\y) arc[start angle=\nd + \ndd, end angle=360 -\nd - \ndd, radius=\srad];
    
    \draw[\c] (top3r\x\y) arc[start angle=\nd, end angle=-\nd, radius=\srad];
    
    \draw[\c] (top3l\x\y) arc[start angle=180 - \ndp + \nddd, end angle=180 + \ndp - \nddd - 1.1, radius=\sbr];
}

\foreach \x/\y/\c/\t in {0/3/\B/$\Dbase$, 1/3/\R/$\Mbase$} {
    \coordinate (center) at ($(\x * \rad * 2 + \x * \xmargin, \y * \rad * 2 -  \rad + \y * \ymargin + 2 * \ymargin)$);
    \node[\c] at (center) {\t};
}

\foreach \y/\t in {0/2, 1/1, 2/0} {
    \draw[->] ($(center0\y) + (\rad + \xmargin * 0.2, 0)$) -- ($(center0\y) + (\rad + \xmargin * 0.8, 0)$);
    \node[label=left:$\t:$] at ($(center0\y) - (\rad + \xmargin *0.3, 0)$) {};
}

\node[\G] at ($(center00) - (\rad / 3.8, 0)$) {$S_2$};

\node[\G] at ($(center01) + (\rad / 2.1, 0)$) {$S_1$};

\node[\R] at ($(center10) - (\rad * 0.75, 0)$) {$A$};

\node[\R] at ($(center10) + (\rad * 0.25, 0)$) {$B$};
\end{tikzpicture}
        \caption{This figure illustrates how adding sets into $\blueDbase$, affects the implicit partition $\redMbase$. Initially, $\blueDbase$ and $\redMbase$ are both $\{\V\}$. \textcolor{Red}{$S_i$} represents sets added to $\blueDbase$ at step $i$. As shown here, the number of sets in $\redMbase$ increases by at least one in each step, since the partial intersections of the newly added \textcolor{red}{$S_i$} with sets in $\redMbase$ create smaller minimal sets. Furthermore, in each step, for any set $S$ in $\closure(\redMbase)$, $\pi(S\mid\blueDbase) = 0$. For example, for family $\mathcal{X} = \{\{A\}, \{B\}, \{A, B\}\}$, we have $\pi^{(2)}(\mathcal{X}) = \pi^{(1)}(\mathcal{X} \mid \{S_2\}) = \pi(\mathcal{X} \mid \{S_1, S_2\}) = 0$.}
        \label{fig:time-complexity}
    \end{figure}
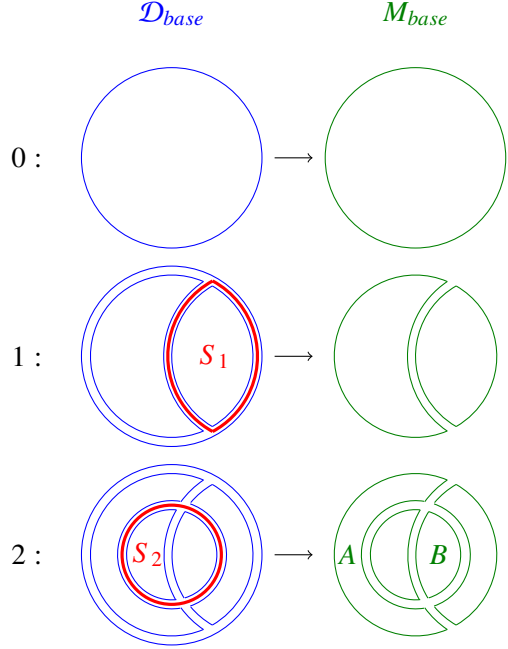

    We now proceed to analyze the recursion depth of the algorithm. After the first recursion step, $\Dbase$ will be non-empty. Then, take $\Mbase$ to be the family of all minimal non-empty sets $S \in \closure(\Dbase)$ such that there exists no other non-empty set $S'\in \closure(\Dbase)$ with $S' \subset S$. We claim that $\Mbase$ is a partition of $\V$. 
    First, any two sets $S, S' \in \Mbase$ cannot have a non-empty intersection, as $S\cap S'$ is in $\closure(\Dbase)$ and would violate the minimality of $S$ and $S'$. 
    Next, let $S$ be the union of the sets in $\Mbase$. 
    As each set in $\Mbase$ is in $\closure(\Dbase)$, $S$ must also be in $\closure(\Dbase)$.
    If $S\neq\V$, $S'=\comp{S}$ would be a non-empty set in $\closure(\Dbase)$ and there must exist a minimal non-empty subset of $S'$ which would be included in $\Mbase$. But $S'$ does not share any elements with any set in $\Mbase$. Therefore, $S'$ must be empty and $S=\V$.
    
    We also note that $\closure(\Mbase)=\closure(\Dbase)$. Since $\Mbase \subseteq \closure(\Dbase)$, it is trivial that
    $\closure(\Mbase) \subseteq \closure(\Dbase)$.
    Now, consider any set $S \in \Dbase$. The intersection of this set with any set $S' \in \Mbase$ must be either $S'$ or the empty set based on the minimality of $S'$. Since $\Mbase$ is a partition of $\V$, $S$ can be written as the union of its intersection with the elements of $\Mbase$, which we have shown to be a union of elements in $\Mbase$. Therefore,  $\Dbase\subseteq\closure(\Mbase)$ and subsequently $\closure(\Dbase)\subseteq\closure(\Mbase)$.
    
    Finally, we show that the number of elements in $\Mbase$ grows in each recursive step. Figure \ref{fig:time-complexity} illustrates how $\Dbase$ and $\Mbase$ change during recursive steps. Let $\Dbase$ and $\Mbase$ be defined as above and $\Dalg$ be the next tight family to be added to $\Dbase$. 
    If the current step is not the final recursive call, then $\pi(\Dalg\mid\Dbase)\neq0$.  
    Then, by Lemma \ref{lm:subadd}, there must be a set $S\in\Dalg$ with $\pi(\{S\}\mid\Dbase)\neq0$. 
    Consider the intersection of $S$ with the sets in $\Mbase$. For some set $S'$ in $\Mbase$, it must be true that $\emptyset\neq S'\cap S \neq S'$. Otherwise, $S$ would be in $\closure(\Mbase)$, and $\pi(\{S\}\mid\Dbase)$ would be equal to zero. 
    Now, $\Mbase'=(\Mbase\setminus \{S'\})\cup \{S\cap S',S'\setminus S\}  $ is a family of non-empty sets in $\closure(\Dbase\cup\Dalg)$ with one more element than $\Mbase$ that partitions $\V$. 
    No minimal set in $\closure(\Dbase\cup\Dalg)$ can contain elements from two distinct sets in $\Mbase'$, as its intersection with either of those sets would contradict its minimality. 
    So, the size of $\Mbase$ in the next step must be at least $\lvert\Mbase'\rvert = \lvert\Mbase\rvert + 1$. As the number of elements in the partition $\Mbase$ cannot grow to more than $\lvert\V\rvert$, the algorithm has a recursive depth of at most $O(\lvert\V\rvert)$. This completes the proof.
\end{proof}


\section{Acknowledgements}
The work is partially supported by DARPA QuICC, ONR MURI 2024 award on Algorithms, Learning, and Game Theory, Army-Research Laboratory (ARL) grant W911NF2410052, NSF AF:Small grants 2218678, 2114269, 2347322.

%
%
%
\bibliographystyle{abbrv}
\bibliography{mybibliography}

\end{document}